\documentclass{article}

\usepackage{amsfonts}
\usepackage{amsmath}
\usepackage{graphicx}
\usepackage{epstopdf}
\usepackage{braket}

\usepackage{Preamble}

\newcommand{\spanv}{{\mathrm{span}}}

\title{\vspace{-1cm}Asymptotic Localization in the Bose-Hubbard Model}

\author{
Alex Bols\\
\normalsize
Instituut voor Theoretische Fysica, K.U.Leuven\\
\normalsize\normalsize
Celestijnenlaan 200 D\\
\normalsize
3001 Leuven, Belgium\\
\normalsize
E-mail: \texttt{alexander.bols@kuleuven.be}
\vspace{0.5cm}\\
Wojciech De Roeck\\
\normalsize
Instituut voor Theoretische Fysica, K.U.Leuven\\
\normalsize
Celestijnenlaan 200 D\\
\normalsize
3001 Leuven, Belgium\\
\normalsize
E-mail: \texttt{wojciech.deroeck@kuleuven.be}\\
}

\date{}

\begin{document}

\maketitle

\abstract{We consider the Bose-Hubbard model.  Our focus is on many-body localization, which was described by many authors in such models, even in the absence of disorder.  Since our work is rigorous, and since we believe that the localization in this type of models is not strictly valid in the infinite-time limit, we necessarily restrict our study to `asymptotic localization': We prove that transport and thermalization are small beyond perturbation theory in the limit of large particle density. Our theorem takes the form of a many-body Nekhoroshev estimate. An interesting and new aspect of this model is the following: even though our analysis proceeds by perturbation in the hopping, the localization can not be inferred from a lack of hybridization between zero-hopping eigenstates. 
}

\section{Introduction}

\subsection{Many-Body Localization and Asymptotic localization}
Many-body localization (MBL) is a  tantalizing  phenomenon, which can be considered a breakdown of basic thermodynamics in strongly disordered quantum spin systems
\cite{basko2006metal,gornyi2005interacting,oganesyan2007localization,vznidarivc2008many,pal2010many,potter2015universal,vosk2015theory,ros2015integrals, khemani2016obtaining,kjall2014many,luitz2015many}.
According to us, the most powerful and useful definition is via the existence of a complete set of quasi-local conserved quantities (aka LIOM's, `local integrals of motion'), see \cite{serbyn2013local,imbrie2016many,huse2014phenomenology}.  This definition stresses that MBL systems are robustly integrable. A lot of interesting characteristics, like absence of transport and thermalization, area-law of entanglement, logarithmic spreading of entanglement, etc. can be deduced from this definition. All of this is discussed at length in the literature, see e.g.\ \cite{nandkishore2014many,eisert2015quantum,imbrie2016review}

Several authors have proposed that MBL should apply as well in systems without quenched disorder, in particular in translation-invariant systems, see \cite{carleo2012localization,de2014asymptotic,de2015asymptotic,
hickey2016signatures,kagan1984localization,schiulaz2015dynamics,schiulaz2013ideal,grover2014quantum,
yao2014quasi,pino2016nonergodic,he2016possibility,kim2016localization,singh2016localization}, as the one studied in this paper.
Whereas this issue is not settled yet, and in particular, we believe that this claim is false when taken literally \cite{de2014scenario,de2016absence} (see also the numerics in \cite{papic2015many,van2015dynamics,bera2016density}),  it is clear that a lot of localization phenomenology is present in such systems. To put it succinctly; they behave for a very long time as if they were localized, even if their eigenvectors might eventually be ergodic (i.e.\ satisfy ETH, see \cite{deutsch1991quantum,srednicki1994chaos}).
In earlier work, we dubbed this property `asymptotic localization' (often one encounters also the term 'quasi-MBL') stressing the fact that transport and thermalization in such models are non-perturbative effects (weaker than any order in perturbation theory) and hence they are visible only on very long timescales. 

There has already been rigorous work establishing this localization phenomenology  \cite{huveneers2013drastic,de2014asymptotic,de2015asymptotic}. It should be noted that, unlike for Anderson localization, rigorous results on MBL are scarce, with the notable exception of \cite{imbrie2016many}, see however also \cite{stolz2016many,seiringer2016decay,abdul2016mathematical,mastropietro2016localization}.
The present paper establishes asymptotic MBL for the Bose Hubbard model, which is a natural model due to the interest in cold atom systems. 
However, the main motivation for this work is that the (asymptotic) localization in the Bose-Hubbard model is far more subtle than in most of the translation-invariant systems referred to above. This is explained next.

%
%
%


\subsection{The `quantum localized'  versus  `classical localized' regime}
In this section, for the sake of focus, we write `localization' without distinguishing between genuine and asymptotic localization, but simply referring to `localization phenomena'.   We also rectrict to chains (dimension one).
Let us start with a general bosonic chain Hamiltonian of the type
\beq \label{eq: basic ham}
H=  \sum_x U N_x^{q} +   J (a_xa^*_{x+1}  + a_{x+1}a^*_{x})
\eeq
with $[a_x,a^*_{x'}]=\delta_{x,x'}$ (bosonic fields) and $N_x=a^*_xa_x$.
For $q=1$, this Hamiltonian describes non-interacting bosons, whereas for $q=2$, we recover the Bose-Hubbard model studied in this paper.
The main idea behind localization in this model applies at large mean density $ \langle N_x \rangle $:  Consider the energy difference $\Delta E$ resulting from one boson hopping from $x$ to $x+1$:
$$
\Delta E =  U\left((N_{x+1}+1)^q+(N_x-1)^q \right)
-  U \left(N^q_{x+1}+N_x^q \right)  \approx Uq (N_{x+1}^{q-1}- N_{x}^{q-1})    $$
Since, in most configurations at large density,  $N_{x+1}-N_x$ is typically large as well, we see that $\Delta E$ is typically increasing with the mean density  as soon as $q>1$. 
The energy difference $\Delta E$ should of course be compared with a typical matrix element of the hopping. The latter is not just of order $J$ because the operators $a_x,a_x^*$ grow with density as well, as $\sqrt{N_x}$, and hence the matrix element for the one-boson hop goes roughly as 
$ g \sqrt{N_xN_{x+1}} 
$. So 
\beq \label{eq: naive hybridization}
\frac{\text{matrix element}}{\Delta E}  \approx  (J/U)  \langle N_x \rangle^{2 - q},
\eeq
which suggests that localization occurs for $q>2$ and the Bose-Hubard model $q=2$ is the critical case, where localization could still be expected for $(J/U) \ll 1$. 

However, what the above reasoning shows, is simply that there is no localization in the $N_x$ basis for $q \leq 2$. That does not exclude localization in some other basis. 
This type of pitfalls is always an issue with heuristic approaches to localization: one argues for localization/delocalization by convergence/divergence of a certain locator expansion, but that is necessarily tied to a choice of basis, which might not be the correct one.  A discussion of examples where this happens can be found e.g.\ in \cite{de2016absence}, the most striking example is the `quantum percolation problem', see \cite{shapir1982localization,veselic2005spectral}, though one might also consider one-particle localization in the continuum as an instance of this.  The above model for $1< q\leq 2$ furnishes in fact a new example of this kind, i.e.\ a case of `non-obvious' localization.

Indeed, let us discuss at least two indications why the no-localization conclusion for $q\leq 2$ is suspicious:
\paragraph{The non-interacting case $q=1$}
 At $q=1$, one has non-interacting bosons. It is then well-known that if $U=U_x$ is a random field, then the one-boson problem  is exponentially localized. If we second-quantize, then obviously the localization persists in an appropriate sense. It is instructive to look at an extreme toy model where the lattice consists of two points $x=1,2$ and we have two eigenmodes $\phi_{1,2} \in \bbC^2$ strongly localized in the sense that 
 $$\str\phi_1(2)\str^2=\str\phi_2(1)\str^2=\epsilon\ll 1,$$
 The second-quantized, localized base states are 
$\str n_{\phi_1},n_{\phi_2}\rangle \equiv  (a^*_{\phi_1})^{n_{\phi_1}}a^*_{\phi_2})^{n_{\phi_2}} \Omega$ where  $a^*_{\phi_i}=\sum_x \phi_i(x) a^*_x$ are the creation operators of the modes $\phi_i$ and $\Omega$ is the vacuum state (no bosons).  The states are rather delocalized in the basis of `bare' number states $\str n_{1},n_{2}\rangle \equiv  (a^*_{1})^{n_{1}}(a^*_{2})^{n_{2}} \Omega$. They are typically spread out over order of $\sqrt{n}$ states, when $n \approx n_1 \approx n_2$.   This means that the hybridization between zero-hopping eigenstates, as predicted from \eqref{eq: naive hybridization} indeed happens, but it does not tell the full story, namely that there is localization in some other basis.
%

\paragraph{Semi-classical limit} The classical problem resembling \eqref{eq: basic ham}, namely the nonlinear Schr\"{o}dinger chain,
\beq \label{eq: schrodinger chain}
H(\psi_x) =  \sum_x  U\str\psi_x\str^{2q}  + J(\psi_x \overline\psi_{x+1} + \overline\psi_x \psi_{x+1})
\eeq
shows `asymptotic localization' all the way down to $q>1$. This was proved in \cite{de2014asymptotic} for $q=2$ but the argument remains valid for all $q>1$, see  \cite{de2014asymptotic} for the relevant heuristics.  
The model \eqref{eq: schrodinger chain} can indeed be viewed as a semi-classical limit of \eqref{eq: basic ham} because at large density, the commutator $[a_x,a^*_x]=1$ is small with respect to the typical values of these operators.  We do not develop this theme, started in \cite{hepp1974classical}, see e.g.\ \cite{benedikter2016mean,hepp1974classical} for references.\\

These two observations suggest that the localization can be exhibited in a basis that is reminiscent of coherent states  (\cf the semi-classical limit alluded to above), in particular different from the basis in which the $N_x$ are diagonal.
 This is indeed what happens.   For the sake of simplicity, the present paper deals with the case $q=2$ only, (Bose Hubbard model) and so we do not focus on this interesting dichotomy between number basis and coherent-like states. Indeed, at the critical case $q=2$, it is the ratio $J/U$ that tunes the hybridization between zero-hopping states. 
 However, the details of the `localization basis' can be read off from the generators of the unitary transformations that approximatively diagonalize the Hamiltonian. The fact that there is hybridization between $N_x$-eigenstates then corresponds simply to the statement that, even locally, these generators have large norm when $J/U \gg 1$, see the discussion around equation \eqref{eq:size of generator} at the end of section \ref{sec_removal of non-resonant terms}.  In fact, in line with the above remarks, the proofs in this paper are quite similar to those in \cite{de2015asymptotic}.
 
\subsection{Many-Body Nekhoroshev estimates}
 
Nekhoroshev estimates \cite{poschel1993nekhoroshev} in classical mechanics express the fact that dynamics is very slow, also away from KAM tori. In this light, our work can be seen as providing a many-body version of these concepts, though there is no exact analogue: Typical Nekhoroshev estimates express that action variables remain close to their original values for long times. We do not prove nor do we believe this to be true in our system. Instead we focus on the energy content $H_I$ of intervals $I$ and we show, see Theorem \ref{thm_nekhoroshev}, that this remains close to its typical value, where the notion of closeness resides in the fact that the difference $H_I(t)-H_I(0)$ is independent of the length $\str I \str$. Moreover, this estimate does not apply deterministically, but only typically with respect to the Gibbs state. Such a restriction is clearly unavoidable in view of the fact that the localization phenomena depend on high density. 
 \subsection{Acknowledgements}

The authors benefit from funding by the InterUniversity Attraction Pole DYGEST (Belspo, Phase VII/18), from the  DFG
(German Research Fund) and the FWO (Flemish Research Fund).



%
%

\section{Model and Results}
We define our model with more care for details and we state the results.
\subsection{Model}

Let $N \geq 1$ be an odd integer and let $\Z_N = \{-(N-1)/2, \cdots, (N-1)/2 \}$. We define the Hilbert space
\begin{equation}
\caH_N := \otimes_{x \in \Z_N} l^2 (\N) \sim l^2 (\N{^{\Z_N}}),
\end{equation}
\ie at each site labelled by $x \in \Z_N$ there is an infinite-dimensional spin-space. For an operator $\hat O$ acting on $\caH_N$ we denote by $s(\hat O)$ the minimal set  $A \subset \Z_N$ such that $\hat O = \hat O_A \otimes \I_{\Z_N \setminus A}$ for some $\hat O_A$ acting on $\caH_A$, and $\I_{A'}$ the identity on $\caH_{A'}$ for any $A' \subset \Z_N$. We do not distinguish between $\hat O_A$ and $\hat O$, and we will denote both by the same symbol.

Let $\hat a, \hat a^*$ be the bosonic annihilation/creation operators on $l^2(\N)$:
\begin{equation}
(\hat a f)(n) = \sqrt{n+1} f(n+1), \hspace{1.5cm} (\hat a^* f)(n+1) = \sqrt{n+1} f(n), \hspace{1.5cm} \text{for } n \in \N.
\end{equation}
We write $\hat a_x, \hat a^*_x$ for the annihilation/creation operators acting on site $x$. We also define the number operators
\begin{equation}
\hat N_x := \hat a^*_x \hat a_x
\end{equation}
and the total particle operator
\begin{equation}
\hat N := \sum_{x \in \Z_N} \hat N_x.
\end{equation}
The vectors diagonalizing the operators $\hat N_x$ play a distinguished role in our analysis. For a finite set $A$ we define the \emph{phase space} $\Omega_A := \N^A$ with elements
\begin{equation}
\eta = (\eta(x))_{x \in A}, \hspace{1.5cm} \eta(x) \in \N
\end{equation}
such that $\caH_A \sim l^2(\Omega_A)$ and we often use $\eta$ as a label for the function $\delta_{\eta}$ \ie $\delta_{\eta}(\eta') = \delta_{\eta, \eta'}$ for $\eta, \eta' \in \Omega_A$. We will also write $\Omega_N := \Omega_{\Z_N}$.

For each odd integer $N$ we consider the Bose-Hubbard Hamiltonian on $\caH_N$ with a chemical potential $\mu$ and free boundary conditions:
\begin{equation}
\hat H = \hat D + g\hat V = \sum_{x \in \Z_N} \hat H_x = \sum_{x \in \Z_N} (\hat D_x + g\hat V_x) 
\end{equation}
with $g \in \R$ a fixed coupling constant and 
\begin{align}
\hat D_x = \hat N_x^2 \;\;\; &\text{ for all } x \in \Z_N, \\
\hat V_x = \hat a_x^* \hat a_{x+1} + \hat a_x \hat a_{x+1}^* \;\;\; &\text{ for } x \in \Z_N \setminus \{(N-1)/2\} \;\;\; \text{ and } \;\;\; \hat V_{(N-1)/2} = 0.
\end{align}
The dimensionless coupling constant $g$ plays to role of the ratio $J/U$ mentioned in the introduction.
\subsection{States and correlations}

All operators appearing in the proof will belong to the algebra $\caB$ finitely generated by the field operators $\hat a_x, \hat a^{\dag}_x$ and bounded functions of the number operators $\hat N_x$. We define states on this algebra corresponding to the grand canonical ensemble at infinite temperature and chemical potential $\mu$ by setting for each operator $\hat O \in \caB$
\begin{equation} \label{def_themal states}
\omega_{\mu} \big( \hat O \big) := \frac{\Tr \ed^{-\mu \hat N} \hat O}{\Tr \ed^{-\mu \hat N}}.
\end{equation}
As was explained in the introduction, we expect localization phenomena in the regime of high occupation numbers. We realise this regime by considering the measure $\omega_{\mu}$ at $\mu \ll 1$.

\subsection{Results}

First we state a result expressing the invariance of the energy content of finite sub-volumes for long times (again polynomial in $\mu^{-1}$). This result is reminiscent of Nekhoroshev estimates for classical systems with a finite number of degrees of freedom. Let $I = \{a_1, a_1 + 1, \cdots, a_2 \} \subset \Z_N$ be a discrete interval and let $\hat H_I = \sum_{x \in I} \hat H_x$. Then
\begin{theorem}\label{thm_nekhoroshev}
For any $n \geq 1$ there is a number $C_n < +\infty$ which is independent of the chain length $N$, such that for any interval $I$ as above we have
\begin{equation}\label{eq:nekhoroshev}
\omega \left( \big( \hat H_I(t) - \hat H_I(0) \big)^2 \right) \leq C_n \mu^{-4} \qquad \text{for any} \quad 0 \leq t \leq \mu^{-n}
\end{equation}
if $\mu$ is sufficiently small.
\end{theorem}

For a thermalizing system, one expects that, during the process of equilibration, the energy content of the interval $I$ changes with time at a rate roughly set by the thermal diffusivity. Once equilibrium is attained after some characteristic time $\tau_{eq}$, the energy in the interval $I$ fluctuates around its equilibrium value, at which point the left hand side of \eqref{eq:nekhoroshev} is expected to scale with the length $\abs{I}$ of the interval. What theorem \ref{thm_nekhoroshev} shows is that the persistent flow of energy into or out of the region $I$ required to equilibrate the system is absent for times that are polynomially long in $\mu^{-1}$. Indeed, we can take the interval $I$ so large that the expected difference of the energy content of the sub-volume $I$ from its equilibrium value is much larger than the bound $\mu^{-4}$ in the right hand side of \eqref{thm_nekhoroshev}. For an equilibrating system we would then expect that at some point in time the bound is surpassed because there should be a persistent energy current until the equilibrium energy content is reached, but the theorem shows that these persistent currents are so small that the bound is not passed at times that are polynomially long in $\mu^{-1}$.  A fortiori, this shows that the timespan $\tau_{eq}$ needed for the system to reach equilibrium goes to infinity faster than any power of $\mu^{-1}$.

Next, we consider in more detail the occurrence of persistent energy currents at long but finite times. For this purpose we first introduce some definitions.
We define the energy current through the bond $(x, x+1)$ for $x \in \Z_N$ by
\begin{equation}\label{def:current}
\hat J_x := i \; \ad_{\hat H} \hat H_{>x} \qquad \text{ where } \qquad \hat H_{>x} := \sum_{y > x} \hat H_y
\end{equation}
\ie $\hat J_x$ is the observable corresponding to the rate of change of the energy contained in the part of the chain to the right of the site $x$, which is indeed the energy current across the bond $(x, x+1)$.

The theorem above is a direct consequences of the following abstract result expressing that, to all orders in perturbation in $\mu$, no persistent energy currents can be produced by the dynamics. Only local oscillations of the energy density can be seen at any finite order in $\mu$.

\begin{theorem}\label{thm_currentdecomposition}
For any integer $n \geq 1$ we can find a number $C_n < +\infty$ independent of $N$ such that the current across any bond $(a, a+1)$ with $a \in \Z$ can be decomposed as
\begin{equation*}
\hat J_{a} = i \; \ad_{\hat H} \hat U_a + \mu^{n+1} \hat G_a
\end{equation*}
if $\mu$ is small enough. The required smallness of $\mu$ is independent of $N$.

The self-adjoint operators $\hat U_a, \hat G_a$ are of zero average, $\omega(\hat U_a) = \omega(\hat G_a) = 0$, and they are local in the sense that they depend only on sites labelled by $z \in \Z_N$ with $\abs{z - a} \leq C_n$, and they satisfy the bounds
\begin{equation}
\omega(\hat U_a^2) \leq C_n \mu^{-4} \; , \qquad \omega(\hat G_a^2) \leq C_n.
\end{equation}
\end{theorem}
The proof of this theorem occupies the largest part of this work.

Let us see how it suggests slowness of transport, even beyond Theorem \ref{thm_nekhoroshev}, in particular out of equilibrium. 

The time-integrated current trough the bond $(a,a+1)$, is, by integrating \ref{thm_currentdecomposition}, given by 
$$
\int_0^t \hat J_a =   \hat U_a(t)-\hat U_a(0)+ \mu^{n+1} \int_0^t \hat G_a.
$$
This seems to be small regardless of the initial state of the system, the point being that the first term on the right-hand side is a temporal boundary term. However, to make this smallness into a mathematical statement, we have to deal with the fact that $U_a$ and  $G_a$ are unbounded operators that can only be controlled on typical states, in casu states of not too high occupation number. Even if we were to assume that the initial state at time $0$ has moderate occupations around $a$ (which suffices because the operators $U_a,G_a$ are local, then we cannot mathematically rule out that this remains so up to time $t$.   However, this sounds very much like a technical problem and it seems reasonable to conclude that 
$$
\str \omega_{\text{non.\ eq}}(\int_0^t \hat J_a) \str \leq C_n \mu^{n+1}t
$$
for almost any reasonable initial state $\omega_{\text{non.\ eq}}$, in particular a non-equilibrium steady state connected to heat reservoirs at different high temperatures (so that the particle density remains high). 

Finally, we note that we formulated our results for energy currents and fluctuations, even though there is another conserved quantity, namely particle number $\sum_x \hat N_x$. Indeed, analogous theorems can  also be formulated for particle number flucutations and currents. 
Also, the restriction to infinite temperature and finite $\mu$ is for simplicity, and we could also have taken finite but large temperature. The only important thing is that the average density remains high $\langle N_x\rangle_{\mu,\beta} \gg 1$.

\section{Outline of the proof}

The proof is inspired by \cite{de2015asymptotic} which in turn picks up on the ideas introduced in \cite{imbrie2016many} .

The first step in the proof is the construction of a change of basis $\hat \Omega$ that almost diagonalizes the Hamiltonian $\hat H$. We want the unitary operator $\hat \Omega$ to be quasi-local in the sense that it is generated by some anti-hermitian matrix $\hat K$ which can be written as a sum of local terms. If we succeed in fully diagonalizing the Hamiltonian in this way, then the transformed Hamiltonian $\hat \Omega^{\dag} \hat H \hat \Omega$ still has a local structure. Moreover, since it is now also diagonal, its eigenstates $\psi_x$ are all localized around some site $x$ in the sense that the decomposition of $\psi_x$ in the basis of perfectly localized states $\delta_y$ has negligible contributions from localized states that are far away from $x$. It would follow that the Hamiltonian $\hat H$ cannot transport energy over a distance longer than the spread of the eigenstates $\psi_x$ around the site $x$, \ie the conductivity would be zero.

Let's try to construct such a change of basis. Note first that the thermal expectation values of the field operators go as $\omega_{\mu}(\hat a_x^{\dag}), \omega_{\mu}(\hat a_x) \sim \mu^{-1/2}$. Hence $\omega_{\mu}(\hat N_x) \sim \mu^{-1}$, $\omega_{\mu}(\hat H_0) \sim \mu^{-2}$ and $\omega_{\mu}(\hat V) \sim \mu^{-1}$. Therefore, if $\mu$ is small, in a typical state we have $\hat H_0 \gg \hat V$ and we can apply perturbation theory for such states. Now, we want to transform away the interaction terms $\hat V = \sum_x \hat a_x^{\dag} \hat a_{x+1} + h.c.$ through the change of basis $\hat \Omega = \ed^{-\hat K}$ and therefore assume that $\hat K$ will be of the same order as $\hat V$. Let us then expand $\hat \Omega^{\dag} \hat H \hat \Omega$ to first order in $\hat V$ and $\hat K$:
\begin{equation}\label{eq:expansion of change of basis}
\hat \Omega^{\dag} \hat H \hat \Omega = \ed^{\hat  K} \big( \hat H_0 + \hat V \big) \ed^{-\hat K} = \hat H_0 + \hat V + [\hat K, \hat H_0] + \text{higher order terms}.
\end{equation} 
We want $\hat K$ to solve the equation
\begin{equation}
\hat V = [\hat H_0, \hat K].
\end{equation}
Lets write $\hat V = \sum_{\rho} \hat V_{\rho}$ where $\hat V_{\rho}$ is a hopping term whose only non-vanishing matrix elements are $\Bra{\eta + \rho} \hat V_{\rho} \Ket{\eta}$. Then the equation is solved by $\hat K = \sum_{\rho} \hat K_{\rho}$ with
\begin{equation}
\langle \eta' , \hat K_{\rho} \eta \rangle =  \frac{\langle \eta' , \hat V_{\rho} \eta \rangle}{\hat H_0(\eta') - \hat H_0(\eta)}.
\end{equation}
This is all fine as long as the denominator is not too small. Since the only non-vanishing matrix elements of $\hat K_{\rho}$ are $\Bra{\eta + \rho} \hat K_{\rho} \Ket{\eta}$, and at small $\mu$ the expectation value of $\hat N_x$ is large we have to consider
\begin{equation}
\hat H_0(\eta + \rho) - \hat H_0(\eta) = \abs{\eta + \rho}^2 - \abs{\eta}^2 = 2 \eta \cdot \rho + \abs{\rho}^2 \approx 2 \eta \cdot \rho.
\end{equation}
If this quantity is too small, we say that $\rho$ is resonant at $\eta$ and we simply set $\Bra{\eta + \rho} \hat K_{\rho} \Ket{\eta} = 0$. In this way we get rid of all non-resonant matrix elements $\Bra{\eta + \rho} \hat V_{\rho} \Ket{\eta}$ to first order. On the other hand, the higher order terms in \eqref{eq:expansion of change of basis} give new local interaction terms in the transformed Hamiltonian beyond the nearest neighbour interactions of the original Hamiltonian. The long range interactions generated in this way are small and therefore we can restrict our attention to those terms which may induce hopping over some finite range $r$. (In the actual proof we truncate the expansion of $\hat \Omega = \ed^{-\hat K}$ and we don't generate interaction terms of arbitrary range at all, but this is a technical point.)

We can now iterate this procedure, at each step trying to transform away the next order non-resonant interactions of the Hamiltonian obtained from the previous change of basis.

We end up with a Hamiltonian
\begin{equation}
\widetilde{H} = \hat H_0 + \sum_{\rho} \widetilde{V}_{\rho}
\end{equation}
where
\begin{equation}
\langle \eta + \rho , \widetilde V_{\rho} \eta \rangle = 0 \qquad \text{ if } \qquad \abs{\eta \cdot \rho} > C.
\end{equation}
and all hopping vectors $\rho$ are supported on sites lying in some interval of length at most $r$ and $\str\rho_x\str \leq r$.
For each such finite-range hopping $\rho$ in $\eta$-space, the condition $\abs{\eta \cdot \rho} > C$ defines a thickened hyperplane in $\eta$-space where transport is possible. Note that the projection on this hyperplane has small Gibbs measure if $\mu$ is small, so in this sense resonances are rare. Since the transformation to the resonant Hamiltonian is local, it is now sufficient to show that the energy current defined by the resonant Hamiltonian is small. This current is given by
\begin{equation}
\widetilde J_a = [\widetilde H, \widetilde H_{>a}]
\end{equation}
where for any $x \in \Z$ we split the resonant Hamiltonian into a left part $\widetilde H_{\leq x}$ and a right part $\widetilde H_{>x}$ by assigning the local terms in $\widetilde H$ either to the left or to the right. The only ambiguity in this splitting is for the terms $\widetilde V_{\rho}$ with $\rho$ a resonance that spans the bond $(x, x+1)$, these terms we can arbitrarily assign to the right part of the splitting.

The probability that there is no resonance across the bond $(a , a+1)$, and thus no current across this bond, is small in $\mu$. But it is only polynomially small in $\mu$ and so we need to do better. The idea is to look for a bond $(x, x+1)$ as close to the bond $(a, a+1)$ as possible, and across which no instantaneous current flows. There are of course very resonant states for which the closest bond without resonance is arbitrarily far from the bond $(a, a+1)$ so we'll have to impose some limit on how far to look for a resonance free bond. Let's say that if we can find a resonance free bond $(x, x+1)$ with $x \in B(a, n)$, then we define $x(\eta) = x$. If there is no resonance free bond $(x, x+1)$ with $x \in B(a, n)$ then we give up and simply put $x(\eta) = a$. Note however that ``giving up'' requires there to be resonance across $2n$ bonds. In the latter scenario we only bound the current across the bond $(a, a+1)$ by a quantity of order one, but the probability of this occurring goes as $\mu^{nc}$ with $c$ some constant of order one. Since $n$ is arbitrary, this is sufficiently improbable for us to bound the Green-Kubo conductivity by an arbitrary power of $\mu$.

Let's try to put some more flesh on this idea. For each $x \in B(a, n)$ we construct an indicator function $\theta_x$ on $eta$-space which equals one if the current through the bond $(x, x+1)$ vanishes and equals zero if it does not. We define $x(\eta)$ to be the $x \in B(a, n)$ such that $\theta_x(\eta) = 1$ which is closest to $a$ if such an $x$ exists, and $x(\eta) = a$ otherwise. In this way we obtain a state-dependent splitting of the Hamiltonian:
\begin{equation}
\widetilde H_L := \sum_{\eta} \widetilde H_{\leq x(\eta)} \hat P_{\eta}, \qquad \widetilde H_R := \sum_{\eta} \widetilde H_{>x(\eta)} \hat P_{\eta}
\end{equation}
and the current becomes
\begin{equation}\label{eq:good expression for J}
\widetilde J_a = [\widetilde H, \widetilde H_{>a}] = [\widetilde H, (\widetilde H_{>a} - \widetilde H_R)] + [\widetilde H, \widetilde H_R].
\end{equation}
The first term is the flow of the local operator $\widetilde H_{>a} - \widetilde H_R$. This operator is sparse in the sense that for most $\eta$ the operators $\widetilde H_{>a}$ and $\widetilde H_R$ are equal, so when we integrate the current over long times we get a small contribution that does not scale with time. The first term is therefore harmless. To analyse the second term we define projectors
\begin{equation}
\hat P_x := \sum_{\eta : x(\eta) = x} \hat P_{\eta}
\end{equation}
so that
\begin{equation}
\widetilde H_R = \sum_{x} \widetilde H_{>x} \hat P_x.
\end{equation}
We then find
\begin{equation}\label{eq:to be made small}
[\widetilde H, \widetilde H_R] = \sum_{x} \left(  [\widetilde H, \widetilde H_{>x}] \hat P_x + \widetilde H_{>x} [\widetilde H, \hat P_x]  \right).
\end{equation}
The states that are not annihilated by the projector $\hat P_x$ have no resonances across the bond $(x, x+1)$ and therefore the commutator $[\widetilde H, \widetilde H_{>x}]$ vanishes on those states, so the first term in the summand is under control. The second term in the summand is not a priori small. It would be small however if the projectors $\hat P_x$ are invariant under the flow, but these projectors are functions of the indicators $\theta_x$. Our goal is therefore to construct indicators $\theta_x$ that at the same time tell us whether there are resonances around the bond $(x, x+1)$ and are (almost) invariant under the flow of $\widetilde H$.

To this end we will investigate the geometry of resonances around the site $x$. Remember that the states at which a given $\rho$ is resonant lie within a thickened hyperplane in $\eta$-space defined by $\abs{\eta \cdot \rho} \leq C$. \ie The hopping vector $\eta$ is resonant in some neighbourhood of the hyperplane $\pi(\rho) := \{ \eta : \eta \cdot \rho = 0 \}$. States that are multi-resonant, say with hopping vectors $\rho_1, \cdots, \rho_p$, are to be found in the neighbourhood of $\pi(\rho_1, \cdots, \rho_p) := \pi(\rho_1) \cap \cdots \cap \pi(\rho_p)$. We will construct sets $\mathsf B(\rho_1, \cdots, \rho_p)$ that contain the $p$-fold resonances around $\pi(\rho_1, \cdots, \rho_p)$ and are invariant under moves $\eta \rightarrow \eta + \rho$ for any hopping vector $\rho \in \spanv\{\rho_1, \cdots, \rho_p\}$.

Let's denote by $P(\eta, E)$ the orthogonal projection of $\eta$ on the subspace $E$. Let $L$ be a large number, we say that $\eta \in \mathsf B (\rho_1, \dots , \rho_p)$ if
\begin{equation*}
\big| \eta - P \big( \eta , \pi(\rho_1, \cdots, \rho_p) \big) \big|_2
\; \le \; C L^p
\end{equation*}
and if, for every linearly independent set of hopping vectors $\rho'_1, \dots , \rho'_{p'} \in \spanv \{ \rho_1, \dots ,\rho_p\}$, 
\begin{equation*}
\big| P \big( \eta , \pi (\rho_1', \cdots, \rho'_{p'}) \big) - P \big( \eta , \pi (\rho_1, \cdots, \rho_p) \big) \big|_2
\; \le \;
C \big( L^p - L^{p'} \big).
\end{equation*}

The purpose of the first condition is clear, the set $\mathsf B (\rho_1, \dots , \rho_p)$ consists of states $\eta$ in a very broad cylinder around $\pi(\rho_1, \cdots, \rho_p)$. In particular, any state that has resonant hopping vectors $\rho_1, \dots , \rho_p$ will be contained in the set $\mathsf B (\rho_1, \dots , \rho_p)$.

But this cylinder is clearly not invariant under resonant moves by $\rho \in \spanv\{\rho_1, \cdots, \rho_p\}$. Indeed, for any set of hopping vectors $\rho'_1, \dots , \rho'_{p'} \in \spanv \{ \rho_1, \dots ,\rho_p\}$, the set $\pi(\rho'_1, \dots , \rho'_{p'})$ intersects the boundary of $\mathsf B (\rho_1, \dots , \rho_p)$. The surface of the cylinder is curved, so around these intersections there are states that can hop from the inside of the cylinder to the outside or vice versa. The remedy is to flatten the surface of the cylinder in the neighbourhood of the intersection, this is the purpose of the second condition defining the set $\mathsf B (\rho_1, \dots , \rho_p)$. This is illustrated in figure \ref{fig:3d flattening} d) where the resonant flow is indicated by green arrows. It is clear that jumps along the green arrows leave the flattened sphere invariant.

By taking $L$ large enough the flattened parts of the the cylinder can be guaranteed to be much longer than the distance of any hopping vector. Furthermore, the flattenings are constructed such that the with of the flattening produced by hopping vectors $\rho'_1, \dots , \rho'_{p'}$ is of order $L^{p'/2}$. \ie higher order resonances give rise to wider flattenings that cut deeper into the cylinder. To see why this is important, consider the construction of $\mathsf B (\rho_1, \rho_2, \rho_3)$ for $\rho_1, \rho_2, \rho_3$ three independent hopping vectors which is illustrated in figure \ref{fig:3d flattening}. First, as in figure a), we peel off thin ring-shaped shells to create the flattenings for each individual hopping vector in $\spanv \{ \rho_1, \rho_2, \rho_3 \}$. This makes the set invariant under all those hoppings except around the places where two or more of such rings intersect as in figure b). But this situation is rectified by the flattening corresponding to sets of two hopping vectors in $\spanv \{ \rho_1, \rho_2, \rho_3 \}$ which is constructed in figure c).

We now construct such a set $\mathsf B(\rho_1, \cdots, \rho_p)$ for each set of linearly independent hopping vectors $\{\rho_1, \cdots, \rho_p\}$ with $p \leq n$ whose supports percolate and such that at least one of the hopping vectors has support near to the site $x$ for which we want to construct the invariant indicator $\theta_x$. We call such sets of hopping vectors clusters around $x$. We then put $\theta_x(\eta) = 0$ if there is a cluster $\{\rho_1, \cdots, \rho_p\}$ around $x$ such that $\eta \in \mathsf B(\rho_1, \cdots, \rho_)$ and $\theta_x(\eta) = 1$ otherwise.

First of all, if $\rho$ is a hopping vector whose support spans the bond $(x, x+1)$, then $\rho$ by itself is a cluster around $x$. So if $\rho$ is resonant at $\eta$, then $\eta \in \mathsf B(\rho)$ hence $\theta_x(\eta) = 0$. If on the other hand $\eta$ has no resonances around the bond $(x, x+1)$ then $\theta_x(\eta) = 1$. We see then that $\theta_x$ indeed tells us whether or not there is a resonance spanning the bond $(x, x+1)$ as was required. It is also clear that the probability with respect to the Gibbs measure of having $\theta_x(\eta) = 0$ is polynomially small in $\mu$.

We turn now to the invariance properties of the function $\theta_x$. Fix $\eta$ and let $\rho$ be resonant at $\eta$. We want to show that $\theta_x(\eta) = \theta_x(\eta + \rho)$. There are three possible scenarios:
\begin{enumerate}
\item There is a cluster $\{\rho_1, \cdots, \rho_p\}$ around $x$ such that $\rho \in \spanv\{\rho_1, \cdots, \rho_p\}$ and $\eta \in \mathsf B(\rho_1, \cdots, \rho_p)$. But the set $\mathsf B(\rho_1, \cdots, \rho_p)$ was constructed specifically to be invariant under hoppings in $\spanv\{\rho_1, \cdots, \rho_p\}$ so $\theta_x(\eta) = \theta_x(\eta+\rho) = 0$.
\item There is a cluster $\{ \rho_1, \cdots, \rho_p \}$ around $x$ such that $\rho \perp \spanv\{ \rho_1, \cdots, \rho_p \}$ and $\eta \in \mathsf B(\rho_1, \cdots, \rho_p)$. Then the set $\mathsf B(\rho_1, \cdots, \rho_p)$ is manifestly invariant under translations in the direction of $\rho$ and so we have invariance.
\item If neither scenario 1 nor scenario 2 is the case, then for each cluster $\{ \rho_1, \cdots, \rho_p \}$ around $x$ such that $\eta \in \mathsf B(\rho_1, \cdots, \rho_p)$ we have neither $\rho \in \spanv\{ \rho_1, \cdots, \rho_p \}$ nor $\rho \perp \spanv\{ \rho_1, \cdots, \rho_p \}$. In this case we must have $p=n$ for if $p < n$ we see that $\rho \cup \{ \rho_1, \cdots, \rho_p \}$ is also a cluster around $x$, and since $\rho$ is a resonance we have $\eta$ close enough to $\pi(\rho, \rho_1, \cdots, \rho_p)$ to conclude that also $\eta \in \mathsf B(\rho, \rho_1, \cdots, \rho_p)$ which contradicts the assumption. So in this scenario we have $\eta \in \mathsf B(\rho_1, \cdots, \rho_n)$, \ie $\eta$ is close to resonance with $n$ linearly independent hopping vectors. The probability of this happening is exponentially small in $n$, it goes to zero as some power of $\mu^n$.
\end{enumerate}
We conclude that the indicators $\theta_x$ are invariant up to a very sparse set. Since the projectors $\hat P_x$ depend only on a finite number these indicators, they are also invariant up to a very sparse set and it follows that the commutator $[\widetilde H, \hat P_x]$ appearing in \eqref{eq:to be made small} has a very small expectation value with respect to the Gibbs measure at small chemical potential.

Let's take stock of what we have achieved. In \eqref{eq:good expression for J} we wrote the current in the form
\begin{equation}
\widetilde J_a = [\widetilde H, \widetilde U_a] + \widetilde G_a
\end{equation}
with $\widetilde U_a$ a local observable whose support is centred on the site $a$ and which has a finite expectation value with respect to the Gibbs measure. The operator $\widetilde G_a$ is also local because it is the difference of two local operators $\widetilde J_a$ and $[\widetilde H, \widetilde U_a]$. Furthermore, the Gibbs expectation value of the operator $\widetilde G_a$ is bounded by $\mu^{cn}$ for some number $c>0$. This is the content of Theorem \ref{thm_currentdecomposition}.

\begin{figure}[!htb]
\begin{center}
\includegraphics[width=0.7\textwidth]{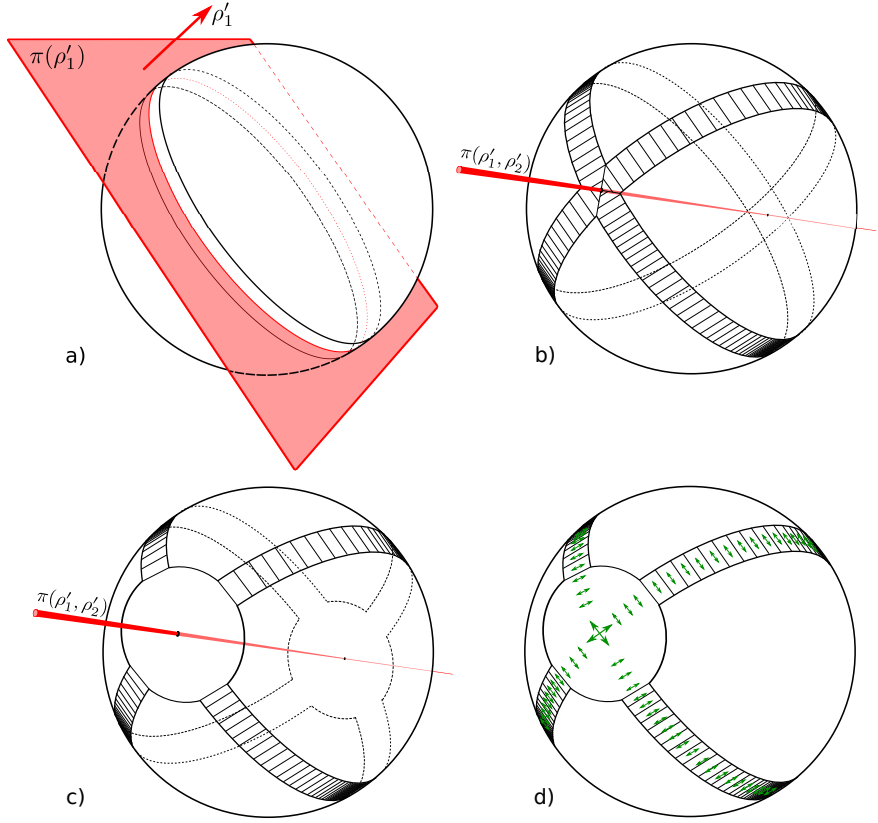}
\end{center}
\caption{Construction of $\mathsf B(\rho_1, \rho_2, \rho_3)$. We start with a sphere of radius $CL^3$ centred on the multi-resonance $\pi(\rho_1m \rho_2, \rho_3)$. In a) we consider a hopping vector $\rho'_1 \in \spanv\{ \rho_1, \rho_2, \rho_3 \}$ and we flatten the sphere where the resonant hyperplane $\pi(\rho'_1)$ intersects the sphere. The width of the flattened strip is of order $L$. The result of doing this for one more hopping vector $\rho'_2$ is shown in b). The flattenings overlap in the neighbourhood of the multi-resonance $\pi(\rho'_1, \rho'_2)$ and it is easily seen that there the flattened sphere is not invariant under hoppings by $\rho'_1$ or $\rho'_2$. This situation is rectified in c) by an extra flattening around the multi-resonance $\pi(\rho'_1, \rho'_2)$. The radius of this flattened disc is of order $L^2$. Finally, the end product is displayed in d) where we also indicated the resonant hopping vectors near the surface of the set in green.}
\label{fig:3d flattening}
\end{figure}

\section{Removal of non-resonant terms from the Hamiltonian} \label{sec_removal of non-resonant terms}

\subsection{Preliminary definitions}

Let $\de \in (0, 1)$ be a number, we define reduced annihilation and creation operators by $\hat \al_x := \de^{1/2} \hat a_x$ and $\hat \al^*_x := \de^{1/2} \hat a^*_x$. We also define reduced number operators by $\hat  n_x := \hat \al^*_x \hat \al_x = \de \hat N_x$.

Suppose $\mu \in (0, 1)$ and consider the Hamiltonian
\begin{equation}
\hat h = \hat d + \mu \hat v = \sum_{x \in \Z_N} \hat h_x = \sum_{x \in \Z_N} (\hat d_x + \mu \hat v_x)
\end{equation}
with
\begin{align}
\hat d_x = \hat n_x^2 \;\;\; &\text{ for all } x \in \Z_N, \\
\hat v_x = g(\hat \al_x^* \hat \al_{x+1} + \hat \al_x \hat \al_{x+1}^*) \;\;\; &\text{ for } x \in \Z_N \setminus \{(N-1)/2\} \;\;\; \text{ and } \;\;\; \hat v_{(N-1)/2} = 0.
\end{align}
By putting $\de = \mu$ and multiplying the reduced Hamiltonian $\hat h$ by $\mu^{-2}$ we recover the Bose-Hubbard Hamiltonian,
but for now we want to think of $\mu$ and $\de$ as independent variables.

Throughout all this work we will deal with operators $\hat f$ in a subspace $\caS$ of the space of linear operators on $\caH$. An operator $\hat f$ belongs to $\caS$ if the following conditions are realized for some number $r = r(\hat f) > 0$.

\begin{enumerate}
\item $\hat f$ can be written as a sum of local terms:
\begin{equation}
\hat f = \sum_x \hat f_x \hspace{0.7cm} \text{with} \hspace{0.7cm} s(\hat f_x) \subset \mathrm B(x, r). 
\end{equation}
where $\mathrm B(x, r) := \{ y \in \Z_N  :  \abs{x-y} \leq r \}$.

\item $\hat f$ has a limited range in the phase space:
\begin{equation}
\hat P_{\eta'} \hat f \hat P_{\eta} = 0 \hspace{0.7cm} \text{whenever} \hspace{0.7cm} \str \eta - \eta' \str_1 > r.
\end{equation}
Also, it follows from spacial locality that $\hat P_{\eta'} \hat f \hat P_{\eta} = 0$ if $\eta - \eta'$ is not supported on a ball of radius $r$.

\item The local terms $\hat f_x$ can be written in the form
\begin{equation}
\hat f_x = \sum_{\substack{\rho \in M_r \\ \supp(\rho) \subset \mathrm B(x, r)}} \hat f_x^{(\rho)} = \sum_{\substack{\rho \in M_r \\ \supp(\rho) \subset \mathrm B(x, r)}} \hat A_{\rho} \hat b_{x, \rho}
\end{equation}
where $\supp(\rho) := \{x \in \Z_N  :  \rho_x \neq 0 \}$ is the support of the vector $\rho$, the set $M_r := \{ \rho \in \Z^N  :  \supp(\rho) \subset \mathrm B(x, r) \text{ for some } x \in \Z_N \text{ and } \max_{x\in Z_N} \abs{\rho_x} \leq r \}$ is the set of moves of range $r$, the operators $\hat A_{\rho}$ are monomials of reduced annihilation and creation operators whose degree is bounded independently of $\delta$ and such that $\hat P_{\eta'} \hat A_{\rho} \hat P_{\eta} = 0$ unless $\eta' - \eta = \rho$ and the operators $\hat b_{x, \rho}$ are diagonal in the number basis.
We also write
\begin{equation}
\hat f^{(\rho)} = \sum_{x} \hat f_x^{(\rho)}
\end{equation}
so that
\begin{equation}
\hat f = \sum_{\rho \in M_r} \hat f^{(\rho)}.
\end{equation}
Note that while the decomposition in local terms may not be unique, this decomposition according to moves is unique.

\end{enumerate}

Let $\xi \in C^{\infty}(\R, [0, 1])$ be a smooth cut-off function : $\xi(-x) = \xi(x)$ for every $x \in \R$, $\xi(x) = 1$ for every $x \in [-1, 1]$ and $\xi(x) = 0$ for every $x \notin [-2, 2]$. For any $a > 0$ we define also $\xi_a$ by $\xi_a(x) = \xi(x/a)$.

Considering the phase space $\Omega_N = \N^N$ as a subset of $\R^N$, we associate to any function $b : \R^N \rightarrow V \subset \C$ a diagonal operator $\hat{b}$ defined by $\hat{b} P_{\eta} = b(\eta) P_{\eta}$ for all $\eta \in \Omega_N$. Conversely, any diagonal operator $\hat{b}$ can be considered associated to some function $b$ on $\R^N$ and this fact will be used without comment. 

For any diagonal operator $\hat{b}$ and any $\rho \in \Z^N$ we define the discrete derivative of $\hat{b}$ in the direction $\rho$ as the diagonal operator $\Delta_{\rho} \hat{b}$ given by 
\begin{equation}
(\Delta_{\rho} \hat{b}) \hat P_{\eta} = \left( b(\eta + \rho) - b(\eta) \right) \hat P_{\eta}
\end{equation}
for all $\eta \in \Omega_N$.

We associate to the operator $\Delta_{\rho} \hat{b}$ a function
\begin{equation}
\Delta_{\rho} b : \R^N \rightarrow \C : \eta \mapsto b(\eta + \rho) - b(\eta).
\end{equation}

Note that the commutator of a product of annihilation and creation operators with a diagonal operator produces a discrete derivative:
\begin{equation}
[\hat A_{\rho}, \hat{b}] = (\Delta_{\rho} \hat{b}) \hat A_{\rho}.
\end{equation}

We define $E(\eta) = \sum_x \eta_x^2$ corresponding to the on-site energy of the Bose-Hubbard Hamiltonian and associated to this function the diagonal operators $\hat{E}$ and $\Delta_{\rho} \hat{E}$ for any $\rho \in \Z^N$ as described above.

Let $1/2 < \gamma < 1$ and define for any $\rho \in \Z^N$ the composite function $\zeta_{\rho} : \R^N \rightarrow [0, 1] : \eta \mapsto \xi_{\delta^{-\gamma}}(\Delta_{\rho}E(\eta))$ and the associated diagonal operator $\hat{\zeta}_{\rho}$.

We define an operator $\caR$ acting on $\caS$ as
\begin{equation}
\caR \hat f = \sum_{\rho \in M_r} \hat f^{(\rho)} \hat{\zeta}_{\rho},
\end{equation}
\ie $\caR \hat f$ consists of the resonant terms of $\hat f$.

Given $\hat f \in \caS$, the equation
\begin{equation}\label{eq_KAM}
\ad_{\hat d} \hat u = \left( \Id - \caR \right) \hat f
\end{equation}
has a solution in $\caS$ given by
\begin{equation}\label{eq_KAMsolution}
\hat u = \delta^{-2} \sum_{\rho \in M_{r(f)}}  \hat f^{(\rho)} \frac{(\I - \hat\zeta_{\rho})}{\Delta_{\rho} \hat E}
\end{equation}
where it is understood that the summand on the right hand side maps $\Ket{\eta}$ to zero if $\Delta_{\rho} E(\eta) = 0$. We will refer to the operator $\hat u$ as \emph{the} solution to equation (\ref{eq_KAM}). It is easily checked that if $\hat f$ is self-adjoint, then $\hat u$ is skew-adjoint.

Given a vector space $X$, a formal power series in $\mu$ is an expression of the form $Y = \sum_{k \geq 0} \mu^k Y^{(k)}$ where $Y^{(k)} \in X$ for every $k \geq 0$. We naturally extend all algebraic operations in $X$ to operations on formal series. Given $l \geq 0$ and given a formal series $Y$, we define the truncation
\begin{equation}
\caT_l(Y) := \sum_{k=0}^{l} \mu^k Y^{(k)}.
\end{equation}
If a formal power series $Y$ is such that $Y^{(k)} = 0$ for all $k > l$ for some $l \in \N$, we will allow ourselves to identify $Y$ with its truncation $\caT_l(Y) \in X$.

\subsection{Perturbative diagonalization}

Given $k \geq 1$, let $\pi(k) \subset \N^k$ be the collection of $k$-tuples $\un{j} = (j_l)_{l=1 \cdots k}$ of non-negative integers satisfying the constraint
\begin{equation}
\sum_{l=1}^k l j_l = k.
\end{equation}
In particular, we have $0 \leq j_l \leq k$.

For $k \geq 0$ we recursively define operators $Q^{(k)}$, $R^{(k)}$ and $S^{(k)}$ on $\caS$, as well as operators $\hat u^{(k)} \in \caS$. Here and below, let us adopt the convention $A^0 = \Id$ for an operator $A$ on $\caS$. We first set $Q^{(0)} = R^{(0)} = \Id$, $S^{(0)} = 0$ and $\hat u^{(0)} = 0$. Next, for $k \geq 1$, we define $\hat u^{(k)}$ as the solution to the equation
\begin{equation}\label{eq_def u}
\ad_{\hat d} \hat u^{(k)} = (\Id - \caR) \left( S^{(k-1)} \hat d + Q^{(k-1)} \hat v \right)
\end{equation}
and then set
\begin{align}
Q^{(k)} &= \sum_{\un{j} \in \pi(k)} \frac{1}{j_1 ! \cdots j_k !} \ad_{\hat u^{(k)}}^{j_k} \cdots \ad_{\hat u^{(1)}}^{j_1} \label{eq_defQ} \\
R^{(k)} &= \sum_{\un{j} \in \pi(k)} \frac{(-1)^{j_1 + \cdots + j_k}}{j_1 ! \cdots j_k !} \ad_{\hat u^{(1)}}^{j_1} \cdots \ad_{\hat u^{(k)}}^{j_k} \label{eq_defR} \\
S^{(k)} &= \sum_{ \substack{ \un{j} \in \pi(k+1) \\ j_{k+1} = 0}  } \frac{1}{j_1 ! \cdots j_k !} \ad_{\hat u^{(k)}}^{j_k} \cdots \ad_{\hat u^{(1)}}^{j_1}. \label{eq_defS}
\end{align}
Note that since $\hat v$ is self-adjoint, $\hat u^{(1)}$ is skew-adjoint so $S^{(1)} \hat d + Q^{(1)} \hat v$ is self-adjoint and it follows in turn that $\hat u^{(2)}$ is skew-adjoint. Continuing in this way we establish inductively that all $\hat u^{(k)}$ are skew-adjoint and all $S^{(k)} \hat d$ and all $Q^{(k)} \hat v$ are self-adjoint operators.
For $n_1 \geq 1$ we define
\begin{equation}\label{eq_defhtilde}
\widetilde{h} = \widetilde{h}_{n_1} = \hat d + \sum_{k=1}^{n_1} \mu^k \caR \left( S^{(k-1)} \hat d + Q^{(k-1)} \hat v \right).
\end{equation}

The following proposition will be shown in subsection \ref{subsec_ProofofProp1} below.

\begin{proposition}\label{prop_resonant hamiltonian}
Let us consider the formal series $R = \sum_{k \geq 0} \mu^k R^{(k)}$ of operators on $\caS$. We have
\begin{enumerate}
\item $\hat h = \caT_{n_1}(R \widetilde{h}_{n_1})$.

\item For every $\hat f = \sum_{k=0}^{n_1} \mu^k \hat f^{(k)}$ with $\hat f^{(k)} \in \caS$ for all $k$, it holds that
\begin{equation}
\ad_{\hat h} (\caT_{n_1}(R \hat f)) = \caT_{n_1} (R \; \ad_{\widetilde{h}_{n_1}} \hat f) + \mu^{n_1 + 1} \; \ad_{\hat v} \sum_{k=0}^{n_1} R^{(n_1 - k)} \hat f^{(k)}.
\end{equation}
\end{enumerate}
\end{proposition}

The resonant Hamiltonian $\widetilde{h}$ and the formal operator $R$ have several characteristics that are good to remember.

\begin{enumerate}
\item Both $\widetilde{h}$ and $R$ are expressed as power series in $\mu$, as is seen from (\ref{eq_defhtilde}) and from the definition of $R$ given in Proposition \ref{prop_resonant hamiltonian}. We introduce also the notation
\begin{equation}
\widetilde{h} = \sum_{k=0}^{n_1} \mu^k \widetilde{h}^{(k)} \hspace{0.3cm} \text{ with } \hspace{0.3cm} \widetilde{h}^{(0)} = \hat d \hspace{0.3cm} \text{ and } \hspace{0.3cm} \widetilde{h}^{(k)} = \caR \left( S^{(k-1)} \hat d + Q^{(k-1)} \hat v \right) \;\;\; \text{ for } \;\;\; k \geq 1.
\end{equation}

\item For each $k \geq 0$, the operator $\widetilde{h}^{(k)}$ is an element of $\caS$, and $R^{(k)}$ is an operator on $\caS$. Let $\hat f = \sum_{x \in \Z_N} \hat f_x \in \caS$ be given. The operators $\widetilde{h}^{k}$ and $R^{(k)} \hat f$ can be decomposed as a sum of local terms with, for example, for $k \geq 1$,
\begin{equation}
\widetilde{h}_x^{(k)} = \caR \left( S^{(k-1)} \hat d_x + Q^{(k-1)} \hat v_x \right) \hspace{0.7cm} \text{ and } \hspace{0.7cm} (R^{(k)} \hat f)_x = R^{(k)} \hat f_x.
\end{equation}
Moreover, we will show in subsection \ref{subsec_ProofofProp1} below, that there exists an integer $r_k$ such that
\begin{equation}\label{eq_boundedrange}
r(\widetilde{h}^{(k)}) \leq r_k \hspace{0.7cm} \text{ and } \hspace{0.7cm} r(R^{(k)} \hat f) \leq r_k + r(f),
\end{equation}
where $r$ is the parameter introduced in the definition of the vector space $\caS$.

\item The operators $\widetilde{h}^{(k)}$ and $R^{(k)}$ depend on $\delta$ if $k \geq 1$. Let thus $k \geq 1$. In what follows, we will use the symbol $\hat A$ for polynomials of reduced annihilation and creation operators, the symbol $b$ to denote smooth, bounded functions on $\R^N$ with bounded derivatives of all orders, and the symbol $\hat{b}$ to denote the diagonal operator associated to a function of the form $\eta \mapsto b(\delta^{\gamma} \eta)$. We refer to diagonal operators $\hat{b}$ of this form as \emph{smooth diagonal operators}.  We will show the following assertions in subsection \ref{subsec_ProofofProp1} below. First, there is an integer $m_k$ such that, given $x \in \Z_N$, $\widetilde{h}_x^{(k)}$ can be expressed as a sum of the type
\begin{equation}\label{eq_deltadependencetildeh}
\widetilde{h}_x^{(k)} = \delta^{-2(k-1) \gamma'} \sum_{j=1}^{m_k} \hat A_{j, x} \hat{b}_{j, x}
\end{equation}
where $0 < \gamma' = 1-\gamma < 1/2$ and such that $s(\hat A_{j, x})$ and $s(\hat{b}_{j, x})$ are subsets of $s(\widetilde{h}_x^{(k)})$ and the bounds on the functions $b_{j, x}$ and its derivatives can be chosen uniformly in $x$.

Second, consider an operator $\hat g \in \caS$ such that $\hat g_x = \hat f_x \hat{b}$ with $\hat f = \sum_{x \in \Z_N} \hat f_x \in \caS$. Then there is an integer $m_{k, g}$ such that $(R^{(k)} \hat g)_x$ can be expressed as a sum of the type
\begin{equation}\label{eq_deltadependenceR}
(R^{(k)} \hat g)_x = \delta^{-2k\gamma'} \sum_{j = 1}^{m_{k, g}} \hat A_{j, x} \hat{b}_{j, x}
\end{equation} 
such that $s(\hat A_{j, x})$ and $s(\hat{b}_{j, x})$ are subsets of $s(\widetilde{h}_x^{(k)})$ and the bounds on the functions $b_{j, x}$ and its derivatives can be chosen uniformly in $x$.
\end{enumerate}

\subsection{Proof of Proposition \ref{prop_resonant hamiltonian} and relations (\ref{eq_boundedrange} - \ref{eq_deltadependenceR}) } \label{subsec_ProofofProp1}

\begin{proofof}[Proposition \ref{prop_resonant hamiltonian}]
Given an operator  $\hat u \in \caS$, a formal transformation, seen as an operator on $\caS$, is defined through
\begin{equation}
\ed^{\mu \ad_{\hat u}} = \sum_{k \geq 0} \frac{\mu^k}{k!} \ad_{\hat u^k}.
\end{equation}

Given a sequence $(\hat u^{(k)})_{k \geq 1} \subset \caS$, that we will later identify with the sequence defined by (\ref{eq_def u}), we construct the formal unitary transformation
\begin{align} 
Q &= \cdots \ed^{\mu^n \ad_{\hat u^{(n)}}} \cdots \ed^{\mu^2 \ad_{\hat u^{(2)}}} \; \ed^{\mu \ad_{\hat u^{(1)}}} = \sum_{j_1 \geq 0, \cdots, j_n \geq 0, \cdots} \frac{\mu^{j_1 + \cdots + n j_n + \cdots}}{j_1 ! \cdots j_n ! \cdots } \left( \cdots \ad_{\hat u^{(n)}}^{j_n} \cdots  \ad_{\hat u^{(1)}}^{j_1}  \right) \\
&= \Id + \sum_{k \geq 1} \mu^k \sum_{\un{j} \in \pi(k)} \frac{1}{j_1 ! \cdots j_k !} \ad_{\hat u^{(k)}}^{j_k} \cdots \ad_{\hat u^{(1)}}^{j_1} = \sum_{k \geq 0} \mu^k Q^{(k)}. \label{eq:change of basis}
\end{align}
The formal inverse of $Q$ is given by
\begin{align*}
R &= \ed^{-\mu \ad_{\hat u^{(1)}}}  \ed^{-\mu^2 \ad_{\hat u^{(2)}}} \cdots \ed^{-\mu^n \ad_{\hat u^{(n)}}} \cdots \\
&= \Id + \sum_{k \geq 1} \mu^k \sum_{\un{j} \in \pi(k)} \frac{(-1)^{j_1 + \cdots + j_k}}{j_1 ! \cdots j_k !} \ad_{\hat u^{(1)}}^{j_1} \cdots \ad_{\hat u^{(k)}}^{j_k} = \sum_{k \geq 0} \mu^k R^{(k)}.
\end{align*}

Let us show the first part of Proposition \ref{prop_resonant hamiltonian}. The operators $Q$ and $R$ are formal inverses of each other, so that, for every $\hat f \in \caS$ such that $\hat f = \caT_{n_1} f$, it holds that
\begin{equation}
\hat f = \caT_{n_1} (R \; \caT_{n_1}(Q \hat f)),
\end{equation}
as can be checked by a direct computation with formal series. We will thus be done if we show that
\begin{equation}\label{eq_done}
\widetilde{h}_{n_1} = \caT_{n_1} (Q \hat h).
\end{equation}
We compute
\begin{equation}
Q \hat h = \sum_{k \geq 0} \mu^k Q^{(k)} (\hat d + \mu \hat v) = \hat d + \sum_{k \geq 1} \mu^k \left( Q^{k} \hat d + Q^{(k-1)} \hat v \right).
\end{equation}
It holds that
\begin{equation}
Q^{(k)} = S^{(k-1)} + \ad_{\hat u^{(k)}} \hspace{0.7cm} \text{ for } \hspace{0.7cm} k \geq 1.
\end{equation}
Since $\ad_{\hat u^{(k)}} \hat d = -\ad_{\hat d} \hat u^{(k)}$ for every $k \geq 1$, and taking now $\hat u^{(k)}$ as defined by (\ref{eq_def u}) we obtain
\begin{equation}
Q \hat h = \hat d + \sum_{k \geq 1} \mu^k \left( S^{(k-1)} \hat d - \ad_{\hat d} \hat u^{(k)} + Q^{(k-1)} \hat v \right) = \hat d + \sum_{k \geq 1} \mu^k \caR \left( S^{(k-1)} \hat d + Q^{(k-1)} \hat v \right).
\end{equation}
From this, we derive (\ref{eq_done}).

Let us then show the second part of Proposition \ref{prop_resonant hamiltonian}. The operators $Q$ and $R$ are formal unitary transformations, inverse of each other. Therefore
\begin{equation}\label{eq_adhR=RadQh}
\ad_{\hat h} R = R \ad_{Q \hat h}
\end{equation}
as a direct but lengthy computation with formal series can confirm. Let us next take $\hat f \in \caS$ such that $\hat f = \caT_{n_1}(\hat f)$. By (\ref{eq_done}), we find that
\begin{equation}
\caT_{n_1} (R \; \ad_{\widetilde{h}} \hat f) = \caT_{n_1} (R \; \ad_{\caT_{n_1} (Q \hat h) } \hat f) = \caT_{n_1} ( R \; \ad_{Q \hat h} \hat f),
\end{equation}
since higher order terms do not contribute due to the overall truncation $\caT_{n_1}$. Therefore, by (\ref{eq_adhR=RadQh}),
\begin{align*}
\ad_{\hat h} ( \caT_{n_1} (R \hat f)) - \caT_{n_1} (R \; \ad_{\widetilde{h}} \hat f) &= \ad_{\hat h} (\caT_{n_1} (R \hat f)) - \caT_{n_1} (R \; \ad_{Q \hat h} \hat f) \\
&= \ad_{\hat h} (\caT_{n_1} (R \hat f)) - \caT_{n_1} (\ad_{\hat h} R \hat f).
\end{align*}
Since $\ad_{\hat h} = \ad_{\hat d} + \mu \; \ad_{\hat v}$, it is finally computed that
\begin{equation}
\ad_{\hat h} (\caT_{n_1} (R \hat f)) - \caT_{n_1} (\ad_{\hat h} R \hat f) = \mu^{n_1 + 1} \; \ad_{\hat v} \sum_{k=0}^{n_1} R^{(n_1 - k)} \hat f^{(k)}.
\end{equation}
This completes the proof.
\end{proofof}

\begin{proofof}[(\ref{eq_boundedrange} - \ref{eq_deltadependenceR})]
Given two operators $\hat f, \hat g \in \caS$ we can decompose $\ad_{\hat f} \hat g$ as a sum of local terms $(\ad_{\hat f} g)_x = \ad_{\hat f} \hat g_x$ so that $r(\ad_{\hat f} \hat g) \leq 2r(\hat f) + r(\hat g)$.
If we write $\hat u = \ad_{\hat d}^{-1} (\Id - \caR) \hat f$ for the solution to
\begin{equation}
\ad_{\hat d} \hat u = (\Id - \caR) \hat f
\end{equation}
given by (\ref{eq_KAMsolution}), then we see that $r ( \ad_{\hat d}^{-1} (\Id - \caR) \hat f ) = r(\hat f)$. From this and (\ref{eq_defQ} - \ref{eq_defS}) we readily deduce (\ref{eq_boundedrange}).

Let us now show (\ref{eq_deltadependencetildeh}) and (\ref{eq_deltadependenceR}). Since we are only interested in tracking the dependence on $\delta$ we may simplify notations as much as possible in the following way. We use symbols $\hat A$ and $\hat{b}$ with the same meanings as in the paragraph where (\ref{eq_deltadependencetildeh}) and (\ref{eq_deltadependenceR}) are stated. Let $n \geq 0$. For $\hat g \in \caS$, we write $\hat g \sim \delta^{-n}$ if $\hat g = \sum_{x \in \Z_N} \hat g_x$ with
\begin{equation}
\hat g_x = \delta^{-n} \sum_j \hat A_{\rho_j} \hat{b}_{j, x}
\end{equation}
with all bounds on the functions $b_{j, x}$ and their derivatives uniform in $x$.

For operators $B$ on $\caS$ we write $B \sim \delta^{-n}$ if for any $\hat f \in \caS$ such that $\hat f \sim \delta^{-m}$ we have $B \hat f \sim \delta^{-n-m}$.

We now observe that if $\hat g \sim \delta^{-n}$ and $\hat f \sim \delta^{-m}$ then $\ad_{\hat g} \hat f \sim \delta^{-(n+m) + \gamma}$. Indeed, $\ad_{\hat g} \hat f$ is a sum whose terms take the following form:
\begin{align*}
\delta^{-(n+m)} [\hat A_{\rho} \hat{b}, \hat A_{\rho'} \hat{b}'] &= \delta^{-(n+m)} \left(  \hat A_{\rho} \hat A_{\rho'} (\Delta_{\rho'} \hat{b}) \hat{b}' + [\hat A_{\rho}, \hat A_{\rho'}] \hat{b}' \hat{b} + \hat A_{\rho'} \hat A_{\rho} (\Delta_{\rho} \hat{b}') \hat{b}  \right) \\
&= \delta^{-(n+m) + \gamma} \sum_{i=1}^{l} \hat A_{i, \rho + \rho'} \hat{b}_i
\end{align*}
for some number $l$. The last step is obtained by noticing that discrete derivatives of the smooth diagonal operators have matrix elements of order $\delta^{\gamma}$ and commutators of monomials of reduced annihilation operators can always be written as $\delta$ times a polynomial in the reduced annihilation and creation operators whose terms all effect the move obtained by summing the moves of the commutants.

The diagonal operators $\hat{b}_i$ are products of smooth diagonal operators and discrete derivatives of smooth diagonal operators and are therefore themselves smooth.

We observe also that if $\hat u$ is the solution to the equation
\begin{equation}
\ad_{\hat d} \hat u = (\Id - \caR) \hat f
\end{equation}
with $\hat f \sim \delta^{-n}$, then since we can take
\begin{equation}
\hat u_x = \delta^{-2} \sum_{\rho \in M_{r(f)}} \hat f_x^{(\rho)} \frac{\I - \hat{\zeta}_{\rho}}{\Delta_{\rho} \hat{E}} = \delta^{-n-2+\gamma} \sum_{\rho \in M_{r(f)}} \sum_j \hat A_{\rho_j} \hat{b}_j  \frac{\I - \hat{\zeta}_{\rho}}{ \delta^{\gamma} \Delta_{\rho} \hat{E}}
\end{equation}
where the operators
\begin{equation}
\hat{b}_j  \frac{\I - \hat{\zeta}_{\rho}}{ \delta^{\gamma} \Delta_{\rho} \hat{E}}
\end{equation}
are smooth diagonal operators, we have $\hat u \sim \delta^{-n-2+\gamma}$.

Remembering that we defined $\gamma' = 1 - \gamma$, let us now establish recursively that for $k \geq 1$, we have
\begin{equation}
Q^{(k-1)},\;\; R^{(k-1)}, \;\; S^{(k-1)} \hat d \sim \delta^{-2(k-1)\gamma'} \hspace{0.3cm} \text{ and } \hspace{0.3cm} \hat u^{(k)} \sim \delta^{-2k\gamma' - \gamma}.
\end{equation}
It is easily checked from the definitions that these relations hold for $k = 1$. Let us see that the claim for $1, \cdots, k \geq 1$ implies the claim for $k+1$.

First, from the definitions (\ref{eq_defQ}) and (\ref{eq_defR}) and the fact that if $\un{j} \in \pi(k)$ then $j_1 + 2 j_2 + \cdots + k j_k = k$ we obtain
\begin{equation}
Q^{(k)}, \;\; R^{(k)} \sim (\delta^{-2k \gamma' - \gamma + \gamma})^{j_k} \cdots (\delta^{-2 \gamma' - \gamma + \gamma})^{j_1} = \delta^{-2 \gamma' \sum_l l j_l} = \delta^{-2k\gamma'}.
\end{equation}

Let us then treat $S^{(k)} \hat d$. We decompose $S^{(k)} \hat d = \sum_{\un{j} \in \pi(k+1)} S_{\un{j}}^{(k)} \hat d$. Fixing a $\un{j} \in \pi(k+1)$ and letting $l \geq 1$ be the smallest integer  such that $j_l \geq 1$ we get for some constant $C(\un{j})$ that
\begin{align*}
S_{\un{j}}^{(k)} &= C(\un{j}) \ad_{\hat u^{(k)}}^{j_k} \cdots \ad_{\hat u^{(l)}}^{j_l - 1} (\ad_{\hat u^{(l)}} d) \\
&= - C(\un{j}) \ad_{\hat u^{(k)}}^{j_k} \cdots \ad_{\hat u^{(l)}}^{j_l - 1} \left( (\Id - \caR) ( S^{(l-1)} \hat d + Q^{l-1} \hat v ) \right) \\
&\sim (\delta^{-2k \gamma'})^{j_k} \cdots (\delta^{-2l \gamma'})^{j_l} \delta^{2l \gamma'} \delta^{-2(l-1)\gamma'} \\
&= \delta^{-2k \gamma'}.
\end{align*}

Finally, we have from the definition (\ref{eq_def u}) and the induction hypothesis that
\begin{equation}
\hat u^{(k+1)} \sim \delta^{-2k \gamma' - 2 + \gamma} = \delta^{-2(k+1)\gamma' - \gamma}
\end{equation}
as required.
\end{proofof}

It is instructive to couple back to the introduction and see that the Bose-Hubbard model is indeed a critical case in the sense described there. \ie we want to see explicitly that the generator of the transformation $Q$ constructed above isn't small in $\mu$. From \eqref{eq:change of basis} we see that the leading contribution to the generator is $\mu \hat u^{(1)}$. Lets write the reduced hopping $\hat v = \sum_{\rho} \hat v^{(\rho)}$ as a sum over nearest neighbour hoppings, it then follows from \eqref{eq_KAMsolution} and \eqref{eq_def u} that
\begin{equation} \label{eq:size of generator}
\hat u^{(1)} = \mu^{-2} \sum_{\rho} \hat v^{(\rho)} \frac{(\I - \hat \zeta_{\rho})}{\Delta_{\rho} \hat E} 
\end{equation}
where we also put $\delta = \mu$ because we are dealing with the Bose-Hubbard model. For a typical state with respect to the Gibbs measure we have $v^{(\rho)} \sim g$ and $\Delta_{\rho} \hat E \sim \mu^{-1}$. Therefore $\mu \hat u^{(1)} \sim g$, \ie the leading contribution to the generator is a constant and is not small in the limit $\mu \rightarrow 0$.

\section{Geometry of Resonances}

Given a point $x\in \Z_N$ and considering the phase space as a subset of $\R^N$, we construct a subset $\mathsf R(x)$ of $\R^N$ with the two following characteristics.  
First, if a point of the phase space does not belong to this set, then the energy current for the Hamiltonian $\widetilde h$ vanishes through the bonds near $x$.
Second, it is approximately invariant under the dynamics generated by $\widetilde{h}$, 
meaning that the commutator of $\widetilde h$ and the projector on the space spanned by states in the set vanishes everywhere except on a subspace spanned by classical states in a subset $\mathsf S(x) \subset \R^N$ which is of small probability with respect to the Gibbs measure. 

The ideas of this Section are best understood visually. 
We hope that figure \ref{fig:3d flattening} will help in that respect.
We let 
\begin{equation}\label{definition du parametre r}
r \; = \; r(n_1) \; = \; \max_{1 \le k \le n_1} r_k, 
\end{equation}
where the numbers $r_k$ are defined in (\ref{eq_boundedrange}).
We let $\delta > 0$ be as in Section \ref{sec_removal of non-resonant terms}.

\subsection{Preliminary definitions}

We recall that, given $\rho \in \Z^{N}$, we denote by $\supp (\rho) \subset \Z^{N}$ the set of points $x$ such that $\rho_x \ne 0$ and we have defined the set $M_r \subset \Z^{N}$ of vectors $\rho = (\rho_x)_{x\in\Z_N}$ such that 
$\max_{x\in\Z_N} |\rho_x| \le r$ and $\supp (\rho) \subset \mathrm B (x, r)$ for some x\ in $\Z_N$.
We write $\abs{\rho}^2_2=\sum_x\str \rho_x \str^2$.
One easily checks that for any $\rho \in M_r$ and $r>1$, we have $\str \rho \str_2\leq 2r^2$ and this will be used without further comment.

Given $x\in\Z^d$, we say that a subset $\{ \rho_1, \dots , \rho_p \} \subset M_r$ is a cluster around $x$ if 
\begin{enumerate}
\item 
the vectors $\rho_1, \dots , \rho_p$ are linearly independent,
\item 
if $p\ge 2$, for all $1 \le i \ne j \le p$, there exist $1 \le i_1, \dots , i_m \le p$ 
such that $i_1=i$, $i_m=j$ and $\supp (\rho_{i_s}) \cap \supp (\rho_{i_{s+1}}) \ne \varnothing$ for all $1 \le s \le m-1$,


\item 
$\supp (\rho_j)\subset \mathrm B(x,4r)$ for some $1 \le j \le p$.  
\end{enumerate}

Finally, given $\rho\in M_r$, we define
\begin{equation*}
\pi (\rho) \; = \; \{ \eta \in \R^{N} : \rho\cdot \eta = 0 \}.
\end{equation*}
Given a subspace $E \subset \R^{N}$, and given $\eta \in \R^{N}$, we denote by $P(\eta, E)$ the orthogonal projection of $\eta$ on the subspace $E$.

\subsection{Approximately invariant resonant zones}

Let $L > 0$, 
let $n_2\ge 1$, 
and let $x\in \Z_N$. 
Let us define two subsets of $\R^{N}$: 
a set $\mathsf R_{\delta,n_2}(x) \subset\R^{N}$ of resonant points, 
and a small set $\mathsf S_{\delta,n_2}(x) \subset\R^{N}$ of ``multi-resonant" points.

To define $\mathsf R_{\delta,n_2}(x)$, let us first define the sets $\mathsf B_\delta (\rho_1, \dots , \rho_p) \subset \R^{N}$, 
where $\{ \rho_1, \dots , \rho_p \}$ is a cluster around $x$.
We say that $\eta \in \mathsf B_\delta (\rho_1, \dots , \rho_p)$ if
\begin{equation}\label{first condition B set}
\big| \eta - P \big( \eta , \pi (\rho_1) \cap \dots \cap \pi (\rho_p) \big) \big|_2
\; \le \; L^p \delta^{-\gamma}
\end{equation}
with $\gamma$ as in section \ref{sec_removal of non-resonant terms}
and if, for every linearly independent $\rho'_1, \dots , \rho'_{p'} \in M_r \cap \spanv \{ \rho_1, \dots ,\rho_p\}$, 
\begin{equation*}
\big| P \big( \eta , \pi (\rho_1') \cap \dots \cap \pi (\rho'_{p'}) \big) - P \big( \eta , \pi (\rho_1) \cap \dots \cap \pi (\rho_p) \big) \big|_2
\; \le \;
\big( L^p - L^{p'} \big) \delta^{-\gamma} .
\end{equation*}
We next define $\mathsf R_{\delta,n_2}(x)$ as the union of all the sets $\mathsf B_\delta (\rho_1, \dots , \rho_p) \subset \R^{N}$ with $p \le n_2$.

We then define $\mathsf S_{\delta,n_2} (x)$ as the set of points $\eta\in \R^{N}$ for which there exists a cluster $\{ \rho_1 , \dots , \rho_{n_2} \}$ around $x$, 
such that $|\rho_j \cdot \eta| \le L^{n_2+1} \delta^{-\gamma}$ for every $1 \le j \le n_2$.

We finally define a smooth indicator function on the complement of $\mathsf R_{\delta,n_2}(x)$ by means of a convolution: 
\begin{equation}\label{smooth indicator good set}
\theta_{x,\delta,n_2} (\eta) 
\; = \;
1 - 
\frac{1}{\Big( \int_\R \xi_{\delta^{-\gamma}} (z) \, \dd z \Big)^{N}} \;
\int_{\R^{N}} \chi_{\mathsf R_{\delta,n_2} (x)} (\eta + \eta') \Big( \prod_{x\in \Z_N} \xi_{\delta^{-\gamma}} (\eta'_x) \Big) \, \dd \eta'. 
\end{equation}

\begin{proposition}\label{prop_invariance resonant set}
Let $n_1$ be given, and so $r(n_1)$ defined by \eqref{definition du parametre r} be fixed as well. 
Let then $n_2 \ge 1$ be fixed. 
The following holds for $L$ large enough and $\delta$ and $\mu$ small enough.
\begin{enumerate}
\item
If
$\theta_{x,\delta,n_2} (\eta) > 0$
then 
$\zeta_{\rho}(\eta) = 0$
for all
$\rho\in M_r$
such that 
$\supp (\rho) \subset \mathrm B(x,4r)$.
\item
$ \big(\ad_{\widetilde{h}} \hat{\theta}_{x, \delta, n_2} \big)  \hat P_{\eta} = 0$
for all 
$\eta \in \Omega_N \subset \R^N$
such that 
$\eta \notin \mathsf S_{n_2} (x)$.
\end{enumerate}
\end{proposition}

\subsection{Proof of Proposition \ref{prop_invariance resonant set}}\label{subs_proof of the invariance of the resonant set}

We start by a series of lemmas. 
The first one simply expresses, in a particular case, that if a point is close to two vector spaces, then it is also close to their intersection. 
The uniformity of the constant $\mathrm C$ comes from the fact that we impose the vectors to sit in the set $M_r$.

\begin{lemma}\label{lem_close to two subspaces implies close to the intersection}
Let $p\ge 1$.
There exists a constant $\mathrm C = \mathrm C (r,p) < + \infty$ such that, 
given linearly independent vectors $\rho_1, \dots , \rho_p,\rho_{p+1} \in M_r$ and given $\eta \in \R^{N}$, it holds that 
\begin{equation*}
\big| \eta - P\big( \eta, \pi (\rho_1) \cap \dots \cap \pi (\rho_{p}) \cap \pi (\rho_{p+1}) \big) \big|_2
\; \le \; \mathrm C \,
\Big( 
|\rho_{p+1} \cdot \eta | + \big| \eta - P\big( \eta, \pi (\rho_1) \cap \dots \cap \pi (\rho_{p}) \big) \big|_2
\Big) .
\end{equation*}
\end{lemma}

\begin{proof}
First, 
\begin{multline}\label{first inequality in close to two subspaces implies close to the intersection}
\big| \eta - P\big( \eta, \pi (\rho_1) \cap \dots \cap  \pi (\rho_{p+1}) \big) \big|_2
\; \le \; 
\big| \eta - P\big( \eta, \pi (\rho_1) \cap \dots \cap \pi (\rho_{p})  \big) \big|_2
\\
+ 
\big| P\big( \eta, \pi (\rho_1) \cap \dots \cap  \pi (\rho_{p+1}) \big) - P\big( \eta, \pi (\rho_1) \cap \dots \cap \pi (\rho_{p})  \big) \big|_2.
\end{multline}
The lemma is already shown if the second term in the right hand side is zero. We further assume this  not to be the case. 
Next, since $\rho_{p+1} \cdot P\big( \eta, \pi (\rho_1) \cap \dots \cap \pi (\rho_{p+1}) \big) = 0$, we obtain
\begin{multline*}
\rho_{p+1} \cdot \eta 
\; = \; 
\rho_{p+1} \cdot \Big( \eta - P\big( \eta, \pi (\rho_1) \cap \dots \cap \pi (\rho_{p})  \big) \Big)
\; + \\ 
\rho_{p+1} \cdot \Big( 
P\big( \eta, \pi (\rho_1) \cap \dots \cap \pi (\rho_{p}) \big) - P\big( \eta, \pi (\rho_1) \cap \dots \cap \pi (\rho_{p+1}) \big)
\Big).
\end{multline*}
This implies
\begin{multline}\label{second inequality in close to two subspaces implies close to the intersection}
\Big| \rho_{p+1} \cdot \Big( 
P\big( \eta, \pi (\rho_1) \cap \dots \cap \pi (\rho_{p}) \big) - P\big( \eta, \pi (\rho_1) \cap \dots \cap \pi (\rho_{p+1}) \big) \Big)
\Big|
\; \le 
\\
|\rho_{p+1} \cdot \eta |
+ 
|\rho_{p+1}|_2
\Big|
\eta - P\big( \eta, \pi (\rho_1) \cap \dots \cap \pi (\rho_{p})  \big) 
\Big|_2.
\end{multline}
The vector
\begin{equation}\label{definition of v in close to two subspaces implies close to the intersection}
v \; = \; \frac{P\big( \eta, \pi (\rho_1) \cap \dots \cap \pi (\rho_{p}) \big) - P\big( \eta, \pi (\rho_1) \cap \dots \cap \pi (\rho_{p+1}) \big) }
{\big| P\big( \eta, \pi (\rho_1) \cap \dots \cap \pi (\rho_{p}) \big) - P\big( \eta, \pi (\rho_1) \cap \dots \cap \pi (\rho_{p+1}) \big) \big|_2}
\end{equation}
is well defined since we have assumed that the denominator in this expression does not vanish. 
The bound \eqref{second inequality in close to two subspaces implies close to the intersection} is rewritten as
\begin{multline*}
\big| P\big( \eta, \pi (\rho_1) \cap \dots \cap \pi (\rho_{p}) \big) - P\big( \eta, \pi (\rho_1) \cap \dots \cap \pi (\rho_{p+1}) \big) \big|_2
\; \le 
\\
\frac{|\rho_{p+1} \cdot \eta |
+ 
|\rho_{p+1}|_2
\Big|
\eta - P\big( \eta, \pi (\rho_1) \cap \dots \cap \pi (\rho_{p})  \big) 
\Big|_2}{|\rho_{p+1} \cdot v|}
\end{multline*}
Inserting this last inequality in \eqref{first inequality in close to two subspaces implies close to the intersection}, we arrive at 
\begin{equation*}
\big| \eta - P\big( \eta, \pi (\rho_1) \cap \dots \cap  \pi (\rho_{p+1}) \big) \big|_2
\; \le \\ 
\frac{|\rho_{p+1} \cdot \eta |}{|\rho_{p+1} \cdot v|}
+ 
\bigg( 1 + \frac{|\rho_{p+1}|_2}{|\rho_{p+1} \cdot v|} \bigg) 
\Big|
\eta - P\big( \eta, \pi (\rho_1) \cap \dots \cap \pi (\rho_{p})  \big) 
\Big|_2 .
\end{equation*}

To finish the proof, it remains to establish that $|\rho_{p+1} \cdot v|$ can be bounded from below by some strictly positive constant, 
where $v$ is given by \eqref{definition of v in close to two subspaces implies close to the intersection}. 
Let us show that
\begin{equation}\label{v is one dimensional}
v = \pm \frac{P \big(\rho_{p+1}, \pi (\rho_1) \cap \dots \cap \pi (\rho_p)\big)}{\big| P \big(\rho_{p+1}, \pi (\rho_1) \cap \dots \cap \pi (\rho_p)\big) \big|_2}.
\end{equation}
We can find vectors $\rho_{p+2}, \dots , \rho_N$ so that $\{ \rho_1, \dots, \rho_N \}$ forms a basis of $\R^N$
and so that every vector $\rho_j$ with $p+2 \le j \le N$ is orthogonal to $\spanv \{ \rho_1, \dots , \rho_{p+1} \}$. 
We express the vector $\eta$ in this basis, $\eta = \sum_{j=1}^N \eta^j \rho_j$, and, from \eqref{definition of v in close to two subspaces implies close to the intersection}, 
we deduce that, for some non-zero constant $R$, we have
\begin{equation*}
v \; = \; R\sum_{j=1}^N \eta^j \Big\{ P\big( \rho_j, \pi (\rho_1) \cap \dots \cap \pi (\rho_{p}) \big) - P\big( \rho_j, \pi (\rho_1) \cap \dots \cap \pi (\rho_{p+1}) \big)  \Big\}. 
\end{equation*}
All the terms corresponding to $1 \le j \le p$ vanish since $\rho_j \perp \pi (\rho_j)$, 
the term $P\big( \rho_{p+1}, \pi (\rho_1) \cap \dots \cap \pi (\rho_{p+1}) \big)$ vanishes for the same reason, 
and all the terms corresponding to $j \ge p+2$ vanish too, as they read in fact $\rho_j - \rho_j = 0$. 
So, the only term left is $R \eta^{p+1}P\big( \rho_{p+1}, \pi (\rho_1) \cap \dots \cap \pi (\rho_{p}) \big)$ and, since $|v|=1$, we arrive at \eqref{v is one dimensional}. 

From \eqref{v is one dimensional} we deduce that 
\begin{equation*}
|v\cdot \rho_{p+1}| \; = \;
\big| P \big(\rho_{p+1}, \pi (\rho_1) \cap \dots \cap \pi (\rho_p)\big) \big|_2 .
\end{equation*}
If $\rho_{p+1} \bot \spanv \{ \rho_1, \dots , \rho_p \}$, then the right hand side just becomes $|\rho_{p+1}|_2$.
This quantity is bounded from below by a strictly positive constant since so is the the norm of any non-zero vector in $M_r$. 
Otherwise, if $\rho_{p+1} \cancel{\bot} \spanv \{ \rho_1, \dots ,\rho_p \}$, 
we know however that the quantity cannot vanish since $\rho_{p+1} \notin \spanv\{ \rho_1, \dots , \rho_p \}$.
Because they are only finitely many vectors $\rho \in M_r$ with the property that $\rho \cancel{\bot} \spanv \{ \rho_1, \dots , \rho_p \}$, we conclude that 
the quantity is bounded from below by a strictly positive constant.
\end{proof}

The next Lemma describes the crucial geometrical properties of the sets $\mathsf B_\delta (\rho_1, \dots ,\rho_p)$ 
that allows to establish the second assertion of Proposition \ref{prop_invariance resonant set}.

\begin{lemma}\label{lem_invariance of the sets B}
Let $\{ \rho_1 , \dots , \rho_p \}$ be a cluster around $x$. 
If, given $\mathrm K < + \infty$, $L$ is taken large enough, then, 
for every $\rho \in M_r \cap \spanv \{\rho_1, \dots ,\rho_p\}$, it holds that if
\begin{equation*}
\eta \in \mathsf B_\delta (\rho_1 , \dots , \rho_p) \quad \text{and} \quad  |\rho \cdot \eta| \; \le \; \mathrm K \delta^{-\gamma}
\end{equation*}
then
\begin{equation*}
\eta + t \rho \in \mathsf B_\delta (\rho_1 , \dots , \rho_p)
\quad \text{as long as} \quad  |t| \le \delta^{-\gamma}. 
\end{equation*}
\end{lemma}

\begin{proof}
To simplify some further expressions, let us define
\begin{equation*}
\eta' = \eta - P (\eta, \pi (\rho_1) \cap \dots \cap \pi(\rho_p)) \in \spanv(\rho_1, \dots , \rho_p).
\end{equation*}
The conditions ensuring that $\eta \in \mathsf B(\rho_1, \dots , \rho_p)$ now simply read
\begin{equation}\label{new formulation special set conditions}
|\eta'|_2 \; \le \; L^p \delta^{-\gamma} \qquad \text{and} \qquad
\big|P\big(\eta', \pi (\rho_1') \cap \dots\cap \pi (\rho'_{p'}) \big)\big|_2 \; \le \; \big( L^p - L^{p'} \big) \delta^{-\gamma} \qquad (p' < p),
\end{equation}
for all linearly independent $\rho_1', \dots , \rho_{p'}' \in M_r \cap \spanv \{\rho_1, \dots ,\rho_p\}$.
The condition $|\rho \cdot \eta| \le \mathrm K \delta^{-\gamma}$ implies $|\rho \cdot \eta'|  \le  \mathrm K \delta^{-\gamma}$.
We need to show that
\begin{align}
&|\eta' + t \rho |_2 \; \le \; L^p \delta^{-\gamma}&
 \text{for} \qquad |t|\le \delta^{-\gamma},
\label{point 1 in proof invariance set B}\\
&\big|P\big(\eta' + t\rho, \pi (\rho_1') \cap \dots\cap \pi (\rho'_{p'}) \big)\big|_2 \; \le \; \big( L^p - L^{p'} \big) \delta^{-\gamma}&
 \text{for} \qquad |t|\le \delta^{-\gamma}.
\label{point 2 in proof invariance set B}
\end{align}

Let us start with \eqref{point 1 in proof invariance set B}:
\begin{align*}
|\eta' + t \rho|_2
\; &\le \;
|\eta' - P(\eta',\pi (\rho))|_2 + |P(\eta',\pi (\rho))|_2 + |t| |\rho|_2
\; \le \;
|\eta' \cdot \rho| + |P(\eta',\pi (\rho))|_2 + |t| |\rho|_2 \\
\; &\le \; 
\mathrm K \delta^{-\gamma} + (L^p - L) \delta^{-\gamma} + 2r^2 \delta^{-\gamma}
\; \le \; L^p \delta^{-\gamma}.
\end{align*}
Here, to get the penultimate inequality, we have used \eqref{new formulation special set conditions} and the hypothesis $\rho \in M_r \cap \spanv \{\rho_1, \dots ,\rho_p\}$,
implying in particular $|\rho|_2 \le 2r^2$, 
while the last inequality is valid for large enough $L$.

Let us next move to \eqref{point 2 in proof invariance set B}.
Let us fix $\rho_1', \dots , \rho'_{p'}$. 
It is seen that, if $\rho \in \spanv \{ \rho_1', \dots , \rho'_{p'} \}$, then  \eqref{point 2 in proof invariance set B} is actually satisfied for all $t \in \R$. 
Let us therefore assume $\rho\notin \spanv \{ \rho_1', \dots , \rho'_{p'} \}$.
We write also $\rho = \rho'_{p'+1}$. 
We will show that, because $ |\rho \cdot \eta| \; \le \; \mathrm K \delta^{-\gamma}$, then in fact 
\begin{equation}\label{thing to show in invariance of the sets B}
\big|P\big(\eta', \pi (\rho_1') \cap \dots\cap \pi (\rho'_{p'}) \big)\big|_2 \; \le \; \big( L^p - L^{p'} - 2r^2 \big) \delta^{-\gamma}.
\end{equation}
Since $|\rho|_2 \le 2r^2$, this will imply  \eqref{point 2 in proof invariance set B}.

To establish \eqref{thing to show in invariance of the sets B}, we start by writing the decompositions 
\begin{align}
|\eta'|_2^2 \; = & \; \big| \eta' - P \big(\eta', \pi (\rho'_1) \cap \dots \cap \pi (\rho'_{p'}) \big) \big|_2^2  + \big| P \big(\eta', \pi (\rho'_1) \cap \dots \cap \pi (\rho'_{p'})\big) \big|_2^2, 
\label{decomp 1 in invariance of the sets B}\\
|\eta'|_2^2 \; = & \; \big| \eta' - P \big(\eta', \pi (\rho'_1) \cap \dots \cap \pi (\rho'_{p'+1})\big) \big|_2^2  + \big| P \big(\eta', \pi (\rho'_1) \cap \dots \cap \pi (\rho'_{p'+1})\big) \big|_2^2.
\label{decomp 2 in invariance of the sets B}
\end{align}
We bound the first term in the right hand side of \eqref{decomp 2 in invariance of the sets B} by applying Lemma \ref{lem_close to two subspaces implies close to the intersection}
and then using \eqref{decomp 1 in invariance of the sets B}:
\begin{align*}
\big| \eta' - P \big(\eta', \pi (\rho'_1) \cap \dots \cap \pi (\rho'_{p'+1})\big) \big|_2^2
\; &\le \; 
\mathrm C \Big( |\rho\cdot \eta'|^2 +  \big| \eta' - P \big(\eta', \pi (\rho'_1) \cap \dots \cap \pi (\rho'_{p'})\big) \big|_2^2 \Big) \\
&\le \;
\mathrm C \Big( 
|\rho\cdot \eta'|^2 + |\eta'|^2_2 - \big| P \big(\eta', \pi (\rho'_1) \cap \dots \cap \pi (\rho'_{p'})\big) \big|_2^2
\Big).
\end{align*}
It may be assumed that $\mathrm C \ge 1$.
Reinserting this bound in \eqref{decomp 2 in invariance of the sets B} yields
\begin{align}
|\eta'|_2^2 \; &\le \; 
\mathrm C \Big( 
|\rho\cdot \eta'|^2 + |\eta'|^2_2 - \big| P \big(\eta', \pi (\rho'_1) \cap \dots \cap \pi (\rho'_{p'})\big) \big|_2^2 \Big)
+ 
\big| P \big(\eta', \pi (\rho'_1) \cap \dots \cap \pi (\rho'_{p'+1})\big) \big|_2^2
\nonumber\\
&\le \; 
\mathrm C \Big( 
\mathrm K^2 \delta^{-2\gamma} + |\eta'|^2_2 - \big| P \big(\eta', \pi (\rho'_1) \cap \dots \cap \pi (\rho'_{p'})\big) \big|_2^2 \Big)
+ 
(L^p - L^{p' + 1})^2 \delta^{-2\gamma}.
\label{inequality that cannot be violated in in invariance of the sets B}
\end{align}
where the hypotheses $|\rho\cdot \eta'| \le \mathrm K \delta^{-\gamma}$ and $\eta \in \mathsf B(\rho_1, \dots , \rho_p)$ have been used to get the last line. 

Let us now show that \eqref{inequality that cannot be violated in in invariance of the sets B} implies \eqref{thing to show in invariance of the sets B} for $L$ large enough.
For this let us write
\begin{align*}
|\eta'|_2 \; &= \; (1 - \mu)^{1/2} L^p \delta^{-\gamma}
\qquad \text{with}\qquad 0 \le \mu \le 1, \\
\big| P \big(\eta', \pi (\rho'_1) \cap \dots \cap \pi (\rho'_{p'})\big) \big|_2 \; &= \; (1 - \nu)^{1/2} (L^p - L^{p'} - r^2) \delta^{-\gamma}
\qquad \text{with}\qquad \nu \le 1
\end{align*}
($\mu > 0$ actually, thanks to the hypothesis $|\rho\cdot \eta'| \le \mathrm K \delta^{-\gamma}$).
Showing \eqref{thing to show in invariance of the sets B} amounts showing $\nu \ge 0$.
With these new notations, inequality \eqref{inequality that cannot be violated in in invariance of the sets B} is rewritten as
\begin{align*}
1 + (\mathrm C - 1) \mu 
\; \le & \;
1 - \frac{2}{L^{p-p'-1}}  \\
&+ \frac{1}{L^{2(p-p'-1)}} + \mathrm C \bigg( 
\frac{\mathrm K^2}{L^{2p}} + \frac{2}{L^{p-p'}} + \frac{2r^2}{L^p} - \frac{1}{L^{2(p-p')}} - \frac{r^4}{L^{2p}} - \frac{2 r^2}{L^{2p-p'}}
\bigg) \\
&+ \mathrm C \nu \bigg( 1 - \frac{L^{p'}+r^2}{L^p} \bigg)^2.
\end{align*}
The left hand side is larger or equal to 1. 
But, when $L$ becomes large, the right hand side is larger or equal to 1 only if $\nu > 0$.
\end{proof}

\begin{lemma}\label{lem_extension B}
Let $\{ \rho_1 , \dots , \rho_p \}$ be a cluster around $x$, 
and let $\rho \in M_r$ be such that $\rho \notin \spanv\{ \rho_1 , \dots , \rho_p \}$, but such that $\{ \rho_1, \dots , \rho_p,\rho \}$ is a cluster. 
If, given $\mathrm K < + \infty$, $L$ is taken large enough, then
\begin{equation*}
\eta \in  \mathsf B (\rho_1 , \dots , \rho_p) \quad \text{and} \quad  |\rho \cdot \eta| \; \le \; \mathrm K \delta^{-\gamma}
\qquad  \Rightarrow \qquad 
\eta \in \mathsf B (\rho_1 , \dots , \rho_p,\rho).
\end{equation*}
\end{lemma}

\begin{proof}
Let us write $\rho = \rho_{p+1}$.
Let $\eta \in  \mathsf B (\rho_1 , \dots , \rho_p)$.
By Lemma \ref{lem_close to two subspaces implies close to the intersection} and by hypothesis, it holds that 
\begin{equation*}
\big| \eta - P \big( \eta, \pi (\rho_1) \cap \dots \cap \pi (\rho_{p+1}) \big) \big|_2
\; \le \;
\mathrm C \, (\mathrm K \delta^{-\gamma} + L^p \delta^{-\gamma}) \; \le \; (L-1)L^p \delta^{-\gamma}
\end{equation*}
if $L$ is large enough. 
Then 
\begin{equation*}
\big| \eta - P \big( \eta, \pi (\rho_1) \cap \dots \cap \pi (\rho_{p+1}) \big) \big|_2
\; \le \;
L^{p+1} \delta^{-\gamma}
\end{equation*}
and, for every $\rho_1', \dots , \rho'_{p'} \in M_r \cap \spanv (\rho_1, \dots ,\rho_{p+1})$, with $p' < p+1$,
\begin{align*}
\big| P \big( \eta, \pi (\rho_1') \cap \dots \cap \pi (\rho_{p'}') \big) - P \big( \eta, \pi (\rho_1) \cap \dots \cap \pi (\rho_{p+1}) \big) \big|_2
\; &\le \;
\big| \eta - P \big( \eta, \pi (\rho_1) \cap \dots \cap \pi (\rho_{p+1}) \big) \big|_2 \\
\; &\le \; \big(L^{p+1} - L^{p}\big) \delta^{-\gamma}
\; \le\; \big(L^{p+1} - L^{p'}\big) \delta^{-\gamma}.
\end{align*}
This shows $\eta \in \mathsf B (\rho_1 , \dots ,\rho_p,\rho)$.
\end{proof}

\begin{proofof}[Proposition \ref{prop_invariance resonant set}]
Let us start with the first claim. 
Let $\rho\in M_r$ be such that $\supp (\rho) \subset \mathrm B(x,4r)$, and let $\eta \in \R^N$ be such that $\theta_x (\eta) > 0$. 
On the one hand, from the definition \eqref{smooth indicator good set} of $\theta_x$, it holds that there exists $\eta' \in \R^N$, with $\max_x |\eta_x'| \le 2 \delta^{-\gamma}$,
such that $\eta + \eta' \notin \mathsf R(x)$. 
On the other hand, since $\supp (\rho) \subset \mathrm B(x,4r)$, we conclude that $\{ \rho \}$ alone is a cluster around $x$ so that, 
if $\eta'' \in \R^N$ is such that 
\begin{equation*}
\frac{|\rho \cdot \eta''|}{|\rho|_2} \; = \; \big| \eta'' - P(\eta'',\pi (\rho)) \big| \; \le \; L \delta^{-\gamma},
\end{equation*}
then $\eta'' \in \mathsf R(x)$.
We thus conclude that $|(\eta + \eta')\cdot \rho| > |\rho|_2 L \delta^{-\gamma} > L \delta^{-\gamma}$, and so that 
\begin{equation*}
|\eta \cdot \rho| \; = \; |(\eta + \eta') \cdot \rho \; - \; \eta' \cdot \rho| \; > \; L \delta^{-\gamma} - 4 \delta^{-\gamma} r^2.
\end{equation*}
It follows that
\begin{equation*}
\abs{\Delta_{\rho} E (\eta)} = \abs{2 \eta \cdot \rho + \abs{\rho}^2} > 2L\delta^{-\gamma} -  8 \delta^{-\gamma} r^2 - 4r^4 > 2 \delta^{-\gamma} 
\end{equation*}
if $L$ is large enough and $\delta$ is small enough. We conclude that $\zeta_{\rho}(\eta) = \xi_{\delta^{-\gamma}} (\Delta_{\rho}E(\eta)) = 0$. 

Let us then show the second part of the Proposition.
Since, by \eqref{eq_defhtilde} and \eqref{eq_boundedrange}, 
the Hamiltonian $\widetilde{h}$ takes the form
\begin{equation*}
\widetilde{h}
\; = \; 
\sum_{\rho \in M_r} \hat A_{\rho} \hat{b}_{\rho} \hat{\zeta}_{\rho},
\end{equation*}
we have
\begin{equation*}
\big( \ad_{\widetilde h} \hat{\th}_x \big) \hat P_{\eta} 
\; = \; 
\sum_{\rho \in M_r} \left[ \hat A_{\rho} \hat{b}_{\rho} \hat{\zeta}_{\rho}, \hat{\theta}_x \right] \hat P_{\eta}
\; = \; 
\sum_{\rho\in M_r} b_{\rho}(\eta) \zeta_{\rho}(\eta) \big(\Delta_{\rho} \theta_x(\eta)\big) \hat A_{\rho} \hat P_{\eta} .
\end{equation*}
It is thus enough to show that for all $\rho \in M_r$ and for all $\eta \notin \mathsf S_{n_2} (x)$ we have
\begin{equation*}
\zeta_{\rho}(\eta) = 0
\quad \text{or} \quad 
\Delta_{\rho} \theta_x (\eta) = 0.
\end{equation*}
Let us thus fix $\eta\in \R^{N}$ and $\rho \in M_r$. We will assume $\zeta_{\rho}(\eta) \neq 0$ and show that $\Delta_{\rho} \theta_x (\eta) = 0$ follows.
Note first that since
\begin{equation*}
\zeta_{\rho} (\eta) = \xi_{\delta^{-\gamma}} \big( \Delta_{\rho} E(\eta) \big) = \xi_{\delta^{-\gamma}} (2 \eta \cdot \rho + \abs{\rho}^2)
\end{equation*}
so $\zeta_{\rho}(\eta) \neq 0$ implies $\abs{\eta \cdot \rho} \leq 2 \delta^{-\gamma}$ if $\delta$ is small enough.
Now, by definition \eqref{smooth indicator good set}, 
we see that $ \Delta_{\rho} \theta_x (\eta) = 0 $ if, for every $\eta'\in \R^{N}$ such that $\max_x|\eta'_x| \le 4\delta^{-\gamma}$, 
it holds that 
\begin{equation*}
\eta + \eta' \in \mathsf R_{n_2} (x)
\qquad \Rightarrow \qquad
\eta + \eta' + t \rho \in \mathsf R_{n_2} (x)
\quad \text{for all $t$ such that $|t| \leq 2r^2$}.
\end{equation*}
Here, the maximal value allowed for $|t|$ is simply the maximal length of of a move $\rho \in M_r$.

We distinguish three cases: In cases 1 and 2 it follows that $\Delta_{\rho} \theta_x (\eta) = 0$ as required, and in case three it follows that $\eta \in S_{n_2}(x)$ so that this case actually does not occur.

1. 
There exists a cluster $\{ \rho_1 , \dots , \rho_p \}$ around $x$, with $p\le n_2$, such that 
$\eta + \eta' \in \mathsf B (\rho_1 , \dots , \rho_p)$ and that 
$\rho\,\bot\,\spanv \{ \rho_1 , \dots ,\rho_p \}$.
It is then seen from the definition of $\mathsf B (\rho_1 , \dots , \rho_p)$ that, 
for every $t\in\R$, $\eta + \eta' + t \rho \in  \mathsf B (\rho_1 , \dots , \rho_p)$. 
Therefore $\eta + \eta' + t \rho \in \mathsf R_{n_2} (x)$ for every $t\in\R$, and we are done.

2.
There exists a cluster $\{ \rho_1 , \dots , \rho_p \}$ around $x$, with $p\le n_2$, such that 
$\eta + \eta' \in \mathsf B (\rho_1 , \dots , \rho_p)$ and that 
$\rho\in\spanv \{ \rho_1 , \dots , \rho_p \}$.
Since $|\rho \cdot \eta | \le 2 \delta^{-\gamma}$ and since $\max_x|\eta_x| \le 4\delta^{-\gamma}$, it holds that $|\rho \cdot (\eta + \eta')| \le (8r^2 + 2) \delta^{-\gamma}$.
Then, by Lemma \ref{lem_invariance of the sets B}, for $|t| \le \delta^{-\gamma}$ we still have $\eta + \eta' + t \rho \in \mathsf B (\rho_1 , \dots , \rho_p)$
if $L$ was chosen large enough. So if we take $\delta$ small enough so that $\delta^{-\gamma} \geq 2r^2$, we are done.

3. 
If neither case 1 or case 2 is realized, then for any cluster $\{ \rho_1 , \dots ,\rho_p \}$ around $x$, with $p\le n_2$ and such that 
$\eta + \eta' \in \mathsf B (\rho_1 , \dots , \rho_p)$, it holds that 
$\rho\,\notin\,\spanv \{ \rho_1 , \dots ,\rho_p \}$, 
and that $\rho\,\cancel{\bot}\,\spanv \{ \rho_1 , \dots , \rho_p \}$. 
Let us see that, since we assume that $\eta \notin \mathsf S_{n_2}(x)$, this case actually does not happen. 
First, for all these clusters, we should have $p=n_2$.
Indeed, otherwise $\{ \rho_1 , \dots , \rho_p, \rho \}$ would form a cluster around $x$ containing $p+1 \le n_2$ independent vectors. 
We would then conclude as in case 2 that  $|\rho \cdot (\eta + \eta')| \le (4r^2 + 2) \delta$,
so that, by Lemma \ref{lem_extension B}, 
$\eta + \eta' \in \mathsf B(\rho_1, \dots ,\rho_p, \rho)$ if $L$ has been chosen large enough. 
This would contradict the assumption ensuring that we are in case 3. 
So $p=n_2$ should hold. 
Writing $\eta'' = \eta + \eta'$, 
we should then conclude from the definition of $\mathsf B (\rho_1 , \dots , \rho_p)$ that, for $1 \le j \le n_2$, 
\begin{equation*}
|\rho_j \cdot \eta''|
\; = \;
|\rho_j|_2
\big| \eta'' - P (\eta'',\pi(\rho_j)) \big|_2
\; \le \; 
|\rho_j|_2
\big| \eta'' - P \big(\eta'',\pi(\rho_1) \cap \dots \cap \pi(\rho_{n_2}) \big) \big|_2
\; \le \; 
|\rho_j|_2 L^{n_2}\delta^{-\gamma}
\end{equation*}
But then 
\begin{equation*}
|\rho_j \cdot \eta|
\; = \; 
|\rho_j \cdot (\eta + \eta') - \rho_j \cdot \eta'|
\; \le \; 
|\rho_j \cdot \eta''| + |\rho_j \cdot \eta'|
\; \le \; 
|\rho_j|_2 L^{n_2}\delta^{-\gamma} + 4 r^2 \delta^{-\gamma}
\; \le \; 
L^{n_2 + 1} \delta^{-\gamma}
\end{equation*}
if $L$ is large enough. 
This would contradict $\eta \notin \mathsf S_{n_2}(x)$.
\end{proofof}

\section{Proof of Theorem \ref{thm_currentdecomposition}}

Let $a \in \Z_N$ be given.

\subsection{New decomposition of the Hamiltonian}

The original splitting of the Hamiltonian in a left and right piece that was used in the definition of the current \eqref{def:current} through the bond $(a, a+1)$ corresponds for the reduced Hamiltonian to the following decomposition:
\begin{equation}\label{original decomposition of the original hamiltonian}
\hat h 
\; = \; 
\hat h_{\le a}^{\mathrm O} + \hat h_{>a}^{\mathrm O}
\; := \; 
\sum_{x \le a} \hat h_x \; + \sum_{x > a} \hat h_x. 
\end{equation}
leading to a reduced current
\begin{equation}\label{def:reduced current}
\hat j_{a,a+1} := i \; \ad_{\hat h} \hat h_{> a}^{\mathrm O} =  \mu^4 \hat J_{a, a+1}.
\end{equation}

We will now obtain a new decomposition of the Hamiltonian that is  equivalent to the one above from the point of view of the conductivity, but leading to an instantaneous current that vanishes for most of the configurations in the Gibbs state at small chemical potential.  
 
Let $n_3 \ge 1$. 
For $x\in \mathrm B(a,n_3)$, we define
\begin{align}
\vartheta_{a,x}
\; &= \;   \frac{1}{\caN}   \bigg(  
\Big(\prod_{y \in \mathrm B(a,n_3)} \theta_y\Big)\,  \delta_{a,x} \, + \, \Big(1 \;\; - \prod_{y \in \mathrm B(a,n_3)} \theta_y \Big) \,  \theta_x    
\bigg), 
\label{definition of vartheta a x}\\[2mm]
\vartheta_{a,*}
\; &= \;    \frac{1}{\caN}    \prod_{y\in\mathrm B(a,n_3)}  (1 - \theta_y).
\label{definition of vartheta a *}
\end{align}
with the normalization factor 
\begin{equation}
\caN \; = \; 
\Big(\prod_{y \in \mathrm B(a,n_3)} \theta_y\Big) \, +  \,  
\Big(1-\prod_{y \in \mathrm B(a,n_3)} \theta_y\Big)\Big(\sum_{x \in \mathrm B(a,n_3)} \theta_x \Big) 
\, + \, 
\prod_{y\in\mathrm B(a,n_3)}  (1 - \theta_y) 
\end{equation}
chosen so that
\begin{equation*}
\sum_{x\in\mathrm B(a,n_3)} \vartheta_{a,x} \; + \; \vartheta_{a,*} \; = \; 1,
\end{equation*}
and satisfying $\caN\geq 1$.   
We then define
\begin{align}
\widetilde{h}_{\le a}
\; &= \; 
\sum_{x\in \mathrm B(a,n_3)} \left( \sum_{y \le x} \widetilde h_y  \right) \hat\vartheta_{a,x}
\; + \;
 \left( \sum_{y \le a} \widetilde h_y \right) \hat\vartheta_{a,*}  , 
\label{h tilde less than a}\\
\widetilde{h}_{> a}
\; &= \; 
\sum_{x\in \mathrm B(a,n_3)} \left(  \sum_{y > x} \widetilde h_y  \right)  \hat\vartheta_{a,x}
\; + \;
\left( \sum_{y > a} \widetilde h_y \right) \hat\vartheta_{a,*}   .
\label{h tilde bigger than a}
\end{align}
It holds that 
\begin{equation*}
\widetilde{h} \; = \; \widetilde{h}_{\le a} + \widetilde{h}_{>a}.
\end{equation*}
By the first point of Proposition \ref{prop_resonant hamiltonian}, we finally define a new decomposition
\begin{equation}\label{new decomposition of the original hamiltonian}
\hat h 
\; = \; 
\hat h_{\le a} + \hat h_{>a}
\; = \; 
\mathcal T_{n_1} (R \widetilde{\hat h}_{\le a}) + \mathcal T_{n_1} (R \widetilde{\hat h}_{> a}).
\end{equation}

\subsection{Definition of $\hat u_a$ and $\hat g_a$}

From the definitions \eqref{original decomposition of the original hamiltonian}, \eqref{def:reduced current} and \eqref{new decomposition of the original hamiltonian}, 
and applying the second point of Proposition \ref{prop_resonant hamiltonian}, we find that
\begin{align}
\hat j_{a,a+1} 
\; &= \; 
i \; \ad_{\hat h} \hat h_{> a}^{\mathrm O} 
\; = \;
i\; \ad_{\hat h} (\hat h_{> a}^{\mathrm O}  - \hat h_{> a}) + i\; \ad_{\hat h} \hat h_{> a} 
\label{redecoupe du courant}\\
\; &= \; 
i \; \ad_{\hat h} (\hat h_{> a}^{\mathrm O}  - \hat h_{> a}) 
+
i\; \mathcal T_{n_1} (R \; \ad_{\widetilde{h}} \widetilde{h}_{> a}) 
+ 
i\; \mu^{n_1+1} \ad_{\hat v} \sum_{k=0}^{n_1} R^{(n_1-k)} \widetilde{h}_{> a}^{(k)}.
\end{align}
Let us call $n_0$ the number $n$ appearing in the statement of the Theorem. 
We define self-adjoint operators
\begin{align}
\hat u_a 
\; &= \; 
\hat h_{> a}^{\mathrm O}  - \hat h_{> a} \; - \; \omega \big(  \hat h_{> a}^{\mathrm O}  - \hat h_{> a} \big) , 
\label{definition of u a}\\
\mu^{n_0+1} \hat g_a 
\; &= \; 
i \; \mathcal T_{n_1} (R \; \ad_{\widetilde{h}} \widetilde{h}_{> a}) 
+ 
i \; \mu^{n_1+1} \ad_{\hat v} \sum_{k=0}^{n_1} R^{(n_1-k)} \widetilde{h}_{> a}^{(k)}.
\label{definition of g a}
\end{align}
Then
\begin{equation}
\hat j_{a,a+1} = i \ad_{\hat h}\hat u_a + \mu^{n_0 + 1} \hat g_a.
\end{equation}
We notice that $\omega(\hat g_a) = 0$ since $\mu^{n_0 + 1} \omega(\hat g_a) = \omega( \ad_{\hat h} \hat h_{>a}) = 0$, by invariance of the Gibbs state.

\subsection{Locality}
Let us show that the operators $\hat u_a$ and $\hat g_a$ are local, meaning that their support consists of sites $z$ with $|z - a| \le \mathrm C_{n_0}$, for some constant $\mathrm C_{n_0} < + \infty$. 
To study $\hat u_a$ we observe that 
\begin{equation*}
\hat h_{>a}^{\mathrm O} - \hat h_{>a} 
\; = \;
- ( \hat h_{\le a}^{\mathrm O} - \hat h_{\le a}). 
\end{equation*}
Let us see that $\hat h_{>a}^{\mathrm O} - \hat h_{>a}$ is supported on sites $z$ with $z \ge a -  (n_3 + (n_2 + 5)r)$.
The operator $\hat h_{>a}^{\mathrm O}$ is supported on sites $z$ with $z \ge a$. 
To analyse $\hat h_{>a}$ defined by \eqref{new decomposition of the original hamiltonian}, 
we first notice that the functions $\vartheta_{a,x}$, with $x\in \mathrm B (a,n_3)$, and $\vartheta_{a,*}$, defined by (\ref{definition of vartheta a x}-\ref{definition of vartheta a *}),
only depend on occupation numbers of sites $z$ with $z \ge a - (n_3 + 4 r + n_2 r)$.
By \eqref{h tilde bigger than a}, the same holds true for $\widetilde{h}_{>a}$, 
since, for any $x\in \Z_N$, the operators $\widetilde h_x$ are supported on sites $z$ with $z\ge x-r$. 
By \eqref{eq_boundedrange}, we conclude that $R\widetilde{h}_{>a}$, and so $\hat h_{>a}$, are supported on sites $z$ with $z \ge a - (n_3 + 6r + n_2 r)$. 
The same holds thus for $\hat h_{>a}^{\mathrm O} - \hat h_{>a}$. 
We can similarly show that $ \hat h_{\le a}^{\mathrm O} - \hat h_{\le a}$ is supported on sites $z$ with $z \le a +   (n_3 + (n_2 + 6)r)$. 
We conclude that $\hat u_a$ defined by \eqref{definition of u a} is supported on sites $z$ with $|z - a| \le  (n_3 + (n_2 + 9)r)$. 

We then readily conclude that $\hat g_a$ is local as well, since, 
going back to \eqref{redecoupe du courant}, we see that $\mu^{n_0+1} g_a$ is the sum of two local functions:
\begin{equation*}
\mu^{n_0+1} \hat g_a 
\; = \; 
i \ad_{\hat h} \hat h_{> a} 
\; = \; 
\hat j_{a,a+1} - i \ad_{\hat h} (\hat h_{>a}^{\mathrm O} - \hat h_{>a}).
\end{equation*}

\subsection{An expression for $\ad_{\widetilde{h}} \widetilde{h}_{>a}$} \label{sec: expression for commutator}

We have
\begin{align}
\ad_{\widetilde h}  \widetilde{h}_{>a}
\; =& \; 
\sum_{x\in \mathrm B(a,n_3)}  \bigg( \ad_{\widetilde h} \sum_{y > x} \widetilde h_y \bigg) \hat\vartheta_{a,x} 
\; + \; 
\bigg( \ad_{\widetilde h} \sum_{y > a} \widetilde h_y \bigg) \hat\vartheta_{a,*} 
\nonumber\\
& \; + \; 
\sum_{x\in \mathrm B(a,n_3)}  \bigg( \sum_{y > x} \widetilde h_y \bigg) \big( \ad_{\widetilde h} \hat\vartheta_{a,x} \big)
\; + \;
\bigg( \sum_{y > a} \widetilde h_y \bigg) \big( \ad_{\widetilde h} \hat\vartheta_{a,*} \big).
\label{first expression for the main commutator}
\end{align}
Let us show that the terms in the first sum in the right hand side vanish, \ie
\begin{equation}\label{crucial cancellation}
\bigg( \ad_{\widetilde h} \sum_{y > x} \widetilde h_y \bigg) \hat\vartheta_{a,x} \; = \; 0
\quad \text{for all} \quad x\in\mathrm B(a,n_3).
\end{equation}
Thanks to the presence of the operator $\mathcal R$ in \eqref{eq_defhtilde}, and thanks to \eqref{eq_boundedrange}, 
we decompose $\widetilde h_x$ as
\begin{equation} \label{eq_splitting of h_x}
\widetilde{h}_x
\; = \;
\hat d_x \;\; + \sum_{\substack{\rho \in M_r \setminus \{0\}: \\ \supp (\rho) \subset \mathrm B(x,r)}} \hat A_{\rho} \hat{b}_{\rho} \hat\zeta_{\rho}.
\end{equation}
We also have
\begin{equation}
\ad_{\widetilde h} \sum_{y > x} \widetilde h_y = \left[ \sum_{z \leq x} \widetilde h_z, \sum_{y > x} \widetilde h_y \right] = \sum_{\substack{ z \leq x < y \\ \abs{z-y} \leq 2r }} \left[ \widetilde h_z, \widetilde h_y \right].
\end{equation}
Fix now $\eta \in \Omega_N$. The claim \eqref{crucial cancellation} is certainly true if $\vartheta_{a, x}(\eta) = 0$. We suppose that this is not the case and proceed to show that then
\begin{equation*}
\bigg( \ad_{\widetilde h} \sum_{y > x} \widetilde h_y \bigg) \hat P_{\eta} = 0.
\end{equation*}
This expression consists entirely of terms of the forms
\begin{equation*}
[\hat d_z, \hat d_y] \hat P_{\eta}, \quad \left[ \hat d_z, \hat A_{\rho} \hat{b}_{\rho} \hat{\zeta_{\rho}} \right] \hat P_{\eta},  \quad \left[ \hat A_{\rho} \hat{b}_{\rho} \hat{\zeta_{\rho}}, \hat d_y \right] \hat P_{\eta} \quad \text{and} \quad \left[ \hat A_{\rho} \hat{b}_{\rho} \hat{\zeta_{\rho}}, \hat A_{\rho'} \hat{b}_{\rho'} \hat{\zeta_{\rho'}} \right] P_{\eta}
\end{equation*}
with $\rho, \rho' \in M_r \setminus \{0\}$ and $\supp(\rho), \supp(\rho') \subset \mathrm B(x, 4r)$. The first of these vanishes because it is a commutator of diagonal operators. The others can all be rearranged so that a factor $\hat \zeta_{\rho} \hat P_{\eta} = \zeta_{\rho}(\eta) \hat P_{\eta}$ with $\rho \in M_r \setminus \{0\}$ and $\supp(\rho) \subset \mathrm B(x, 4r)$ appears. We will thus be done if we can show that $\zeta_{\rho}(\eta) = 0$ for all such $\rho$.

But we have $\vartheta_{a, x}(\eta) > 0$ hence $\theta_x(\eta) > 0$. It then follows from the first point of Proposition \ref{prop_invariance resonant set} that $\zeta_{\rho}(\eta) = 0$ for all $\rho \in M_r \setminus \{0\}$ and $\supp(\rho) \subset \mathrm B(x, 4r)$ and so we are done.

The commutator \eqref{first expression for the main commutator} is now rewritten as
\begin{align}
\ad_{\widetilde h}  \widetilde{h}_{>a}
\; =& \; 
\left[ \sum_{a-r \leq x \leq a} \widetilde h_x, \sum_{a < y \leq a+r} \widetilde h_y \right] \hat \vartheta_{a, *}
\nonumber\\
& \; + \; 
\sum_{x\in \mathrm B(a,n_3)}  \bigg( \sum_{x < y \leq a + n_3} \widetilde h_y   \bigg)  \big( \ad_{\widetilde h} \hat\vartheta_{a,x} \big)
\; + \;
\bigg( \sum_{a < y \leq a + n_3} \widetilde h_y \bigg) \big( \ad_{\widetilde h} \hat\vartheta_{a,*} \big).
\label{expression for ad h tilde after cancellation}
\end{align}

\subsection{Definition of an exceptional set $\mathsf Z\subset \Omega$}
  
Let $\mathsf Z  \subset \R^N$ be such that $\eta \in \mathsf Z$ if and only if
there exists $n_2$ linearly independent vectors $\rho_1, \dots , \rho_{n_2} \in M_r$ such that $( \cup_j \supp (\rho_j) ) \subset \mathrm B (a,2n_3)$
and such that $|\rho_j \cdot \eta | \le L^{n_2 + 1} \delta^{-\gamma}$ for $1\le j \le n_2$. 

We also consider for any $s > 0$ a broadening of the set $\mathsf Z$:
\begin{equation}
\mathsf Z_s := \{ \eta \in \Omega_N \;\; : \;\; B(\eta, s) \cap \mathsf Z \neq \emptyset \}.
\end{equation}

\begin{lemma}\label{lem_exceptional set Z}
If, given $n_1$ and $n_2$, the numbers $L$ and $n_3$ have been taken large enough, then
\begin{enumerate}
\item
If $ad_{\widetilde h}  \widetilde{h}_{>a} \hat P_{\eta} \ne 0$, then $\eta \in \mathsf Z$. 
\item
There exists a constant $\mathrm C = \mathrm C(L,n_1,n_2,n_3, s) < + \infty$ such that 
$\omega(\hat P_{\mathsf Z_s})  \le \mathrm C \mu^{n_2}\delta^{-\gamma n_2}$.
\end{enumerate}
\end{lemma}

\begin{proof}
Let us start with the first point.
From \eqref{expression for ad h tilde after cancellation}, we conclude that,
if $ \big( \ad_{\widetilde h}  \widetilde{h}_{>a} \big) \hat P_{\eta} \ne 0$, 
then at least one of the following quantities needs to be non-zero: 
$\vartheta_{a,*}  (\eta)$, 
or $\big( \ad_{\widetilde h} \hat\vartheta_{a,*} \big) \hat P_{\eta}$, 
or $\big( \ad_{\widetilde h} \hat\vartheta_{a,x} \big) \hat P_{\eta}$ for some $x\in \mathrm B(a,n_3)$. 
But if $\ad_{\widetilde h} \big( \prod_{x \in \mathrm B(a, n_3)} (1 - \hat{\theta}_x) \big) \hat P_{\eta} \neq 0$ then $\big( \prod_{x \in \mathrm B(a, n_3)} (1 - \hat\theta_x) \big) \hat P_{\eta'} \neq 0$ for some $\eta' \in \Omega_N$ with $\abs{\eta - \eta'}_2 \leq r$.
Therefore, by inspection of the definitions \eqref{definition of vartheta a x} and \eqref{definition of vartheta a *}, the condition $\big( \ad_{\widetilde h}  \widetilde{h}_{>a} \big) \hat P_{\eta} \ne 0$ implies actually 
\begin{equation*}
\Big( \prod_{x \in \mathrm B(a, n_3)} (1 - \hat\theta_x) \Big) P_{\eta'} \neq 0 \quad \text{ for some } \quad \eta' \in \Omega_N \quad \text{ with } \quad \abs{\eta - \eta'}_2 \leq r
\end{equation*}
or
\begin{equation*}
\big( \ad_{\widetilde h} \hat\theta_x \big) \hat P_{\eta} \neq 0 \quad \text{ for some } x \in \mathrm B(a, n_3).
\end{equation*}

In the first case, $\theta_x(\eta') < 1$ for all $x \in \mathrm B(a, n_3)$. But then, by definition \eqref{smooth indicator good set}, there exists some $\rho \in \R^N$ with $\max_y \abs{\rho_y} \leq 2 \delta^{-\gamma}$ such that $\eta'' = \eta' + \rho \in \mathrm R(x)$. There exists therefore a cluster $\{\rho_1, \dots , \rho_p \}$ around $x$, with $p \leq n_2$, such that \eqref{first condition B set} holds. 
This implies 
\begin{equation*}
|\rho_1 \cdot \eta''|
\; = \;
|\rho_1|_2
\big| \eta'' - P (\eta'',\pi(\rho_1)) \big|_2
\; \le \; 
|\rho_1|_2
\big| \eta'' - P \big(\eta'',\pi(\rho_1) \cap \dots \cap \pi(\rho_{p}) \big) \big|_2
\; \le \; 
|\rho_1|_2
L^{p}\delta^{-\gamma}
\end{equation*}
and therefore $|\eta' \cdot \rho_1| \le |\rho_1|_2 L^{p}\delta^{-\gamma} + r^2 \delta^{-\gamma}$ and hence
\begin{equation*}
\abs{\eta \cdot \rho_1} = \abs{ \eta' \cdot \rho_1 + (\eta - \eta') \cdot \rho_1 } \leq \abs{\rho_1}_2 L^p \delta^{-\gamma}  + r^2 \delta^{-\gamma} + r \leq L^{n_2 + 1} \delta^{-\gamma}
\end{equation*}
if $L$ is large enough, and using that $p \leq n_2$.

It holds by definition of a cluster around $x$ that $\supp(\rho_1) \subset \mathrm B(x,4r)$.   Let us now take another $x'$ such that $\theta_{x'}(\eta') < 1$ and $\str x-x'\str >4r$. Then the same reasoning gives a vector $\rho'_1 \neq \rho_1$ satisfying again  $|\eta \cdot \rho'_1|  \le L^{n+1} \delta^{-\gamma}$.   By taking   $n_3$ large enough, we can find $n_2$ linearly independent vectors and thus  guarantee that $\eta \in \mathsf Z$.

Suppose now that $\ad_{\widetilde h} \hat\theta_{x} \hat P_{\eta} \ne 0$ for some $x\in \mathrm B(a,n_3)$.  
It then follows from the second assertion of Proposition \ref{prop_invariance resonant set} that $\eta \in \mathsf S(x)$, 
so that, by definition, there exists a cluster $\{ \rho_1, \dots , \rho_{n_2} \}$ around $x$ such that $|\eta \cdot \rho_j| \le L^{n_2 + 1} \delta^{-\gamma}$. 
This implies $\eta \in \mathsf Z$.

We now move to the second claim of the Lemma.

Since Gibbs measure factorises in the number basis we find
\begin{align}
\omega \big( \hat P_{\mathsf Z_s}) &= \frac{  \sum_{\eta \in \Omega_N}  \chi_{\mathsf Z_s}(\eta)  \ed^{-\mu \abs{\eta}_1}  }{ \sum_{\eta \in \Omega_N}  \ed^{-\mu \abs{\eta}_1}   }  =  \frac{  \sum_{\eta \in \mu \Omega_N}  \chi_{\mathsf Z_s}(\mu^{-1} \eta)  \ed^{-\abs{\eta}_1  }}{ \sum_{\eta \in \mu\Omega_N}  \ed^{-\abs{\eta}_1}   } \\
&\leq   C(n_3)  \frac{  \int   \chi_{\mathsf Z_s}(\mu^{-1}\eta)  \prod_{x \in B(a, 2n_3)} \ed^{-\eta_x}  \dd \eta_x  }{   \int   \prod_{x \in B(a, 2n_3)} \ed^{-\eta_x}  \dd \eta_x   }
\end{align}
for $\mu$ small enough. It then follows by exploiting that the set $\mathsf Z_s$ is a finite union of cylinders whose bases have volume of order $\delta^{-\gamma n_2}$ that
\begin{equation}
\omega\big( \hat P_{\mathsf Z_s} \big) \leq C \big( \mu \delta^{-\gamma} \big)^{n_2}.
\end{equation}
\end{proof}

\subsection{Bounds on the norms of $\hat u_a$ and $\hat g_a$}

In this subsection, we fix $\de = \mu$, so we are dealing again with the Bose-Hubbard Hamiltonian up to an overall factor of $\mu^2$. For $\hat f \in \caS$ we again use the notation $\hat f \sim \delta^{-n}$ with the same meaning as in the proof of \eqref{eq_deltadependencetildeh} and \eqref{eq_deltadependenceR}.

Let us first look at $\hat u_a$. Using \eqref{eq_deltadependencetildeh} and \eqref{eq_deltadependenceR} and the fact that the functions $\vartheta_{a, x}$ and $\vartheta_{a, *}$ are bounded we find that the local function $\hat h^O_{>a} - \hat h_{>a}$ takes the form
\begin{align}
\hat h^O_{>a} - \hat h_{>a}
\; &= \; \sum_{n=0}^{n_1} \mu^{n} (\hat h^O_{> a} - \hat h_{> a})^{(n)} \label{decomposition  h original - h in powers of epsilon} \\
\; &= \; 
\hat h^O_{>a} - \mathcal T_{n_1} \big( R \widetilde h_{>a}\big)
\; = \;
\hat h^O_{> a} - \sum_{n=0}^{n_1} \mu^n \sum_{k=0}^n R^{(n-k)} \widetilde h^{(k)}_{>a}
\; \sim \;
\sum_{n=0}^{n_1} \mu^{n} \delta^{-2n\gamma'}. \label{decomposition  H original - H differently}
\end{align}
with the dominant contribution for each $n$ coming from the $k=0$ term.

But we can still improve on this for the $(\hat h^O_{>a} - \hat h_{>a})^{(0)}$ term. Define $\mathsf W \subset \Omega$ as the set containing all $\eta$ such that $\theta_x(\eta) <1$ for some $x \in \mathrm B(a, n_3)$. By inspection of the definitions \eqref{definition of vartheta a x} and \eqref{definition of vartheta a *} we have
\begin{equation*}
\vartheta_{a,x}(\eta)=\delta_{a,x}, \qquad  \vartheta_{a,*}(\eta)=0, \qquad   \text{for}\, \,  \eta \in \Omega \setminus \mathsf W,
\end{equation*}
so writing
\begin{align*}
(\hat h^O_{> a} - \hat h_{>a})^{(0)} & = 
\sum_{x >a} \hat d_x -  \widetilde h^{(0)}_{>a}  \\
 &=   \sum_{x>a} \hat d_x -   \sum_{x\in \mathrm B(a,n_3)}   \left( \sum_{y >x}  \hat d_y \right)  \hat\vartheta_{a,x} -  \left( \sum_{y > a} \hat d_y \right)  \hat\vartheta_{a,*}
\end{align*}
we see that
\begin{equation*}
(\hat h^O_{>a} - \hat h_{>a})^{(0)} \hat P_{\eta} = 0, \qquad  \text{for}\, \,  \eta \in \Omega \setminus \mathsf W.
\end{equation*}
Therefore
\begin{equation}\label{eq:bounding u}
\omega \left( \left(  (\hat h^O_{>a} - \hat h_{>a})^{(0)}   \right)^2 \right) = \omega \left( \left(  (\hat h^O_{>a} - \hat h_{>a})^{(0)}   \right)^2 \hat P_{\mathsf W} \right) \leq \omega \left( \left(  (\hat h^O_{>a} - \hat h_{>a})^{(0)}   \right)^4 \right) \omega \big( \hat P_{\mathsf W} \big).
\end{equation}
The operator $\big((\hat h^O_{>a} - \hat h_{>a})^{(0)}\big)^4$ is a finite polynomial of reduced field operators and diagonal operators that are uniformly bounded independent of $\mu$ and $\delta$. Its expectation value with respect the the Gibbs measure can therefore be bounded by some constant. Concerning the projector $\hat P_{\mathsf W}$ we note that for any $\eta \in \mathsf W$ there is at least one $\rho \in M_r$ with $\supp(\rho) \subset \mathrm B(a, n_3)$ and such that $\abs{\eta \cdot \rho} \leq L^{n_2 + 1} \delta^{-\gamma}$ and thus $\omega (\hat P_{\mathsf W}) \leq C \mu \delta^{-\gamma} = \mu^{\gamma'}$. Therefore
\begin{equation}\label{eq:bound on h^(0)}
\omega \left( \left(  (h^O_{>a} - h_{>a})^{(0)}   \right)^2 P_{\mathsf W} \right) \leq C \mu^{\gamma'}
\end{equation}
for some constant $C$ depending on $n_1$ but, crucially, not depending on the volume thanks to the locality of $(h^O_{>a} - h_{>a})^{(0)}$.

In a similar way, from \eqref{decomposition  H original - H differently} we obtain
\begin{equation}\label{eq:bound on h^(n)}
\omega \left( \left( \mu^n (\hat h^O_{>a} - \hat h_{>a})^{(n)} \right)^2 \right) \leq C \mu^{2n} \delta^{-2n \gamma'} = C \mu^{2n(1-\gamma')}
\end{equation}
for $1 \leq n \leq n_1$. Using (\eqref{eq:bound on h^(0)}, \eqref{eq:bound on h^(n)}) in \eqref{decomposition  h original - h in powers of epsilon}  and since $\gamma' \in (0, 1/2)$ we therefore find
\begin{equation}
\omega \left( (\hat h^O_{>a} - \hat h_{>a})^2 \right) \leq  C \mu^{\gamma'}.
\end{equation}
The constant $C$ depends on $n_1$ but not on the volume, and the bound holds for $\mu$ small enough. We conclude that $\omega \big( \hat u_a^2 \big) \leq C \mu^{\gamma'}$

Next, we want $\omega(\hat g_a^2) \leq C\mu^8$ so that $\omega(\hat G_a^2) = \omega((\mu^{-4} \hat g_a)^2) = \mu^{-8} \omega((\hat g_a)^2) $ can be bounded by a constant. We start from the definition \eqref{definition of g a} and we note that both terms in this definition are local;
for $\mathcal T_{n_1} (R \ad_{\widetilde{h}} \widetilde{h}_{> a}) $ this follows from the explicit expression in Section \ref{sec: expression for commutator} 
and the other term is then local as a difference of local terms. 
We compute
\begin{equation}\label{second brol preuve theoreme prinicpal}
\omega(\hat g_a^2)
\; \le \; 
2\mu^{-2 (n_0 + 1)} \omega \left( \left(\mathcal T_{n_1}(R \ad_{\widetilde h} \widetilde h_{> a}) \right)^2 \right)
\; + \; 
2\mu^{2(n_1 - n_0)} \omega \left( \left( \ad_{\hat v} \sum_{k=0}^{n_1} R^{(n_1-k)} \widetilde{h}_{> a}^{(k)} \right)^2 \right).
\end{equation}

We look first at the second term and conclude, by means of \eqref{eq_deltadependencetildeh} and \eqref{eq_deltadependenceR} that it is of the form
\begin{equation*}
\ad_{\hat v} \sum_{k=0}^{n_1} R^{(n_1-k)} \widetilde{h}_{> a}^{(k)}
\; \sim \; 
\delta^{-2n_1 \gamma' + \gamma}.  
\end{equation*}
We can thus bound, using locality, the fact that $\gamma' < 1/2$ and taking $\mu$ small enough,
\begin{equation*}
\mu^{2(n_1 - n_0)} \omega\left( \left( \ad_{\hat  v} \sum_{k=0}^{n_1} R^{(n_1-k)} \widetilde{h}_{> a}^{(k)} \right)^2 \right) \leq C_{n_0} \mu^{8}
\qquad \text{if} \qquad 
n_1 = \left\lceil\frac{n_0 + 4 - \gamma}{1 - 2\gamma'}\right\rceil.
\end{equation*}

We then analyse the first term in \eqref{second brol preuve theoreme prinicpal}. By the first point of Lemma \ref{lem_exceptional set Z}, the expression $\ad_{\widetilde h} \widetilde h_{> a} \hat P_{\eta}$ vanishes for all $\eta \in \Omega_N \setminus \mathsf Z$, so, remembering the definition \eqref{definition du parametre r} of the parameter $r$ we see that  $\mathcal T_{n_1} (R \;  \ad_{\widetilde h} \widetilde h_{> a}) \hat P_{\eta}$ certainly vanishes for all $\eta \in \Omega_N \setminus \mathsf Z_r$ with
\begin{equation}
\mathsf Z_r = \{ \eta \in \Omega_N \; : \; \mathrm B(\eta, r) \cap \mathsf Z \neq \emptyset \}.
\end{equation}
so
\begin{equation}\label{eq:bounding g}
\omega \left(  \left(  \mathcal T_{n_1} (R \;  \ad_{\widetilde h} \widetilde h_{> a}) \right)^2  \right) = \omega \left(  \left(  \mathcal T_{n_1} (R \;  \ad_{\widetilde h} \widetilde h_{> a}) \right)^2 \hat P_{\mathsf Z_r} \right).
\end{equation}
Using (\ref{eq_deltadependencetildeh}, \ref{eq_deltadependenceR}) and $\gamma' < 1/2$, we find that
\begin{equation*}
\mathcal T_{n_1} (R \;  \ad_{\widetilde h} \widetilde h_{> a}) = \sum_{n=0}^{n_1} \mu^n \sum_{m = 0}^n R^{(n-m)} \sum_{j = 0}^m \left[ \widetilde h^{(m-j)}, \widetilde h^{(j)}_{> a}  \right] \sim \sum_{n = 1}^{n_1} \mu^n \delta^{-2n\gamma' + 4\gamma' + \gamma} \sim \mu^{1 + 2\gamma' + \gamma}
\end{equation*}
where we have used the fact that $[\widetilde h^{(0)}, \widetilde h^{(0)}_{> a}] = 0$ to start the last sum at $n = 1$. We conclude that
\begin{equation*}
\omega \left(  \left(  \mathcal T_{n_1} (R \;  \ad_{\widetilde h} \widetilde h_{> a}) \right)^4  \right) \leq C \mu^{4 + 8\gamma' + 4\gamma}
\end{equation*}
hence
\begin{equation}
\omega \left(  \left(  \mathcal T_{n_1} (R \;  \ad_{\widetilde h} \widetilde h_{> a}) \right)^2  \right) \leq \omega \left(  \left(  \mathcal T_{n_1} (R \;  \ad_{\widetilde h} \widetilde h_{> a}) \right)^4  \right) \omega \big( \hat P_{\mathsf Z_r} \big) \leq C \mu^{4(1+\gamma) + (8 + n_2)\gamma'}
\end{equation}
Where we used the second point of \ref{lem_exceptional set Z}. By taking $n_2 = \lceil (10 + 2n_0) / \gamma' \rceil$ we conclude from \eqref{second brol preuve theoreme prinicpal} that $\omega(g_a^2) \leq C_{n_0} \mu^8$ if $\mu$ is small enough.

Thus, if we choose $\gamma' = 1/4$ then we have shown that
\begin{equation}
\omega(u_a^2) \leq C \mu^{1/4} \qquad \text{ and } \qquad \omega(g_a^2) \leq \mu^8.
\end{equation}

\subsection{Back to the original Bose-Hubbard Hamiltonian}

Defining the operators
\begin{equation*}
\hat U_a := \mu^{-2} u_a \qquad \text{and} \qquad G_a := \mu^{-4}g_a 
\end{equation*}
we decompose the current across the bond $(a, a+1)$ as
\begin{align*}
J_{a, a+1} &= i \; \ad_{\hat H} \big( \hat H_{> a} \big) = i \; \ad_{\mu^{-2} \hat h} \big( \mu^{-2} \hat h^O_{>a} \big) = \mu^{-4} i \; \ad_{\hat h} \big( \hat h^O_{>a} \big) \\
&=  \mu^{-4} \hat j_{a, a+1} = i \; \ad_{\mu^{-2} \hat h} \big( \mu^{-2} \hat u_a \big) + \mu^{n_0 + 1} \big( \mu^{-4} \hat g_a \big) = i \; \ad_{\hat H} \hat U_a + \mu^{n_0 + 1} \hat G_a.
\end{align*}
Furthermore, we have showed in the previous subsection that
\begin{equation*}
\omega(\hat U_a^2) = \mu^{-4} \omega(\hat u_a^2) \leq C \mu^{-4} \mu^{1/4} < C \mu^{-4}
\end{equation*}
and
\begin{equation*}
\omega(\hat G_a^2) = \mu^{-8} \omega(\hat g_a^2) \leq C_{n_0}.
\end{equation*}
Taking $n = n_0$ concludes the proof of Theorem \ref{thm_currentdecomposition}.  \hfill$\square$\vskip.5cm

\section{Proof of Theorem \ref{thm_nekhoroshev}}

Let $I = \{a_1+1, \cdots, a_2\} \subset \Z_N$ be an interval on the chain and let $\hat H_I = \sum_{x \in I} \hat H_x$ be the energy in the interval $I$. Since $\hat H_I = \hat H_{> a_1} - \hat H_{> a_2}$ we have
\begin{equation*}
i\ad_{\hat H} \big( \hat H_I \big) = i \; \ad_{\hat H} \big( \hat H_{> a_1} \big) - i \; \ad_{\hat H} \big( \hat H_{> a_2} \big) = \hat J_{a_1, a_1 + 1} - \hat J_{a_2, a_2 + 1}.
\end{equation*}
Let us write $\hat H_I(t)$ for the energy in the interval $I$ at time $t$ in the Heisenberg picture, then
\begin{equation*}
\hat H_I(t) - \hat H_I(0) = \int_0^t \dd s \; \frac{\dd \hat H_I(t)}{\dd t} = \int_0^t \dd s \; i \; \ad_{\hat H} \big(\hat H_I \big)(s) = \int_0^t \dd s \; \Big( \hat J_{a_1, a_1 +1}(s) - \hat J_{a_2, a_2 +1}(s) \Big).
\end{equation*}
Using Theorem \ref{thm_currentdecomposition} we find for any $a \in \Z_N$ that
\begin{equation*}
\int_0^t \dd s \; \hat J_{a}(s) = \int_0^t \dd s \; \Big( i \; \ad_{\hat H} (\hat U_a)(s) + \mu^{n + 1} \hat G_a(s) \Big) = \hat U_a(t) - \hat U_a(0) + \mu^{n+1} \int_0^t \dd s \; \hat G_a(s),
\end{equation*}
so
\begin{equation*}
\hat H_I(t) - \hat H_I(0) = \sum_{j \in \{1, 2\}} (-1)^j \left( \hat U_{a_j}(t) - \hat U_{a_j}(0) + \mu^{n+1} \int_0^t \dd s \; \hat G_{a_j}(s) \right).
\end{equation*}
It follows that
\begin{align*}
\omega \left( \big(\hat H_I(t) - \hat H_I(0) \big)^2 \right) \leq 2 \omega &\left( \left( \sum_{j \in \{1, 2\}} (-1)^j \big( \hat U_{a_j}(t) - \hat U_{a_j}(0) \big) \right)^2 \right) \\
&+ 2 \mu^{2n + 2} \omega \left( \left( \sum_{j \in \{1, 2\}} (-1)^j \int_0^t \dd s \; \hat G_{a_j}(s)  \right)^2 \right).
\end{align*}
Using Cauchy-Schwarz, time-translation invariance of the Gibbs state and the bound of Theorem \ref{thm_currentdecomposition}, the first term in this expression is bounded by
\begin{align*}
\sum_{j, k \in \{1, 2\}} (-1)^{j+k} &\Big( \omega \big( \hat U_{a_j}(t) \hat U_{a_k}(t) \big) - \omega \big( \hat U_{a_j}(t) \hat U_{a_k}(0) \big) - \omega \big( \hat U_{a_j}(0) \hat U_{a_k}(t) \big) + \omega \big( \hat U_{a_j}(0) \hat U_{a_k}(0) \big)  \Big) \\
&\leq 4 \sum_{j, k \in \{1, 2\}} \big( \omega(\hat U_{a_j}^2) \omega(\hat U_{a_k}^2) \big)^{1/2}  \leq C \mu^{-4}.
\end{align*}
Taking $0 \leq t \leq \mu^{-n}$, we can bound the second term in the same way by
\begin{align*}
4 \mu^{2n + 2} \sum_{j \in \{1, 2\}} \omega \left( \left( \int_0^t \dd s \; \hat G_{a_j}(s) \right)^2 \right)  &= 4 \mu^{2n + 2} \sum_{j \in \{1, 2\}} \int_0^t \dd s \int_0^t \dd s' \omega\big(\hat G_{a_j}(s) \hat G_{a_j}(s')\big) \\
& \leq 4 \mu^{2n + 2} \sum_{j \in \{1, 2\}} \mu^{-2n} \omega\big(\hat G_{a_j}^2\big) \leq C \mu^2,
\end{align*}
thus proving the theorem. \hfill$\square$\vskip.5cm

\bibliographystyle{plain}

\bibliography{loclibrary}

\begin{thebibliography}{10}

\bibitem{abdul2016mathematical}
Houssam Abdul-Rahman, Bruno Nachtergaele, Robert Sims, and G{\"u}nter Stolz.
\newblock Mathematical results on many-body localization in the disordered xy
  spin chain.
\newblock {\em arXiv preprint arXiv:1610.01939}, 2016.

\bibitem{basko2006metal}
DM~Basko, IL~Aleiner, and BL~Altshuler.
\newblock Metal--insulator transition in a weakly interacting many-electron
  system with localized single-particle states.
\newblock {\em Annals of physics}, 321(5):1126--1205, 2006.

\bibitem{benedikter2016mean}
Niels Benedikter, Marcello Porta, and Benjamin Schlein.
\newblock Mean-field regime for bosonic systems.
\newblock In {\em Effective Evolution Equations from Quantum Dynamics}, pages
  7--16. Springer, 2016.

\bibitem{bera2016density}
Soumya Bera, Giuseppe De~Tomasi, Felix Weiner, and Ferdinand Evers.
\newblock Density propagator for many-body localization: finite size effects,
  transient subdiffusion,(stretched-) exponentials.
\newblock {\em arXiv preprint arXiv:1610.03085}, 2016.

\bibitem{carleo2012localization}
Giuseppe Carleo, Federico Becca, Marco Schir{\'o}, and Michele Fabrizio.
\newblock Localization and glassy dynamics of many-body quantum systems.
\newblock {\em Scientific Reports}, 2:243, 2012.

\bibitem{de2014asymptotic}
Wojciech De~Roeck and Fran{\c{c}}ois Huveneers.
\newblock Asymptotic quantum many-body localization from thermal disorder.
\newblock {\em Communications in Mathematical Physics}, 332(3):1017--1082,
  2014.

\bibitem{de2014scenario}
Wojciech De~Roeck and Fran{\c{c}}ois Huveneers.
\newblock Scenario for delocalization in translation-invariant systems.
\newblock {\em Physical Review B}, 90(16):165137, 2014.

\bibitem{de2015asymptotic}
Wojciech De~Roeck and Fran{\c{c}}ois Huveneers.
\newblock Asymptotic localization of energy in nondisordered oscillator chains.
\newblock {\em Communications on Pure and Applied Mathematics},
  68(9):1532--1568, 2015.

\bibitem{de2016absence}
Wojciech De~Roeck, Francois Huveneers, Markus M{\"u}ller, and Mauro Schiulaz.
\newblock Absence of many-body mobility edges.
\newblock {\em Physical Review B}, 93(1):014203, 2016.

\bibitem{deutsch1991quantum}
JM~Deutsch.
\newblock Quantum statistical mechanics in a closed system.
\newblock {\em Physical Review A}, 43(4):2046, 1991.

\bibitem{eisert2015quantum}
Jens Eisert, Mathis Friesdorf, and Christian Gogolin.
\newblock Quantum many-body systems out of equilibrium.
\newblock {\em Nature Physics}, 11(2):124--130, 2015.

\bibitem{gornyi2005interacting}
IV~Gornyi, AD~Mirlin, and DG~Polyakov.
\newblock Interacting electrons in disordered wires: Anderson localization and
  low-t transport.
\newblock {\em Physical review letters}, 95(20):206603, 2005.

\bibitem{grover2014quantum}
Tarun Grover and Matthew~PA Fisher.
\newblock Quantum disentangled liquids.
\newblock {\em Journal of Statistical Mechanics: Theory and Experiment},
  2014(10):P10010, 2014.

\bibitem{he2016possibility}
Rong-Qiang He and Zheng-Yu Weng.
\newblock On the possibility of many-body localization in a doped mott
  insulator.
\newblock {\em Scientific Reports}, 6, 2016.

\bibitem{hepp1974classical}
Klaus Hepp.
\newblock The classical limit for quantum mechanical correlation functions.
\newblock {\em Communications in Mathematical Physics}, 35(4):265--277, 1974.

\bibitem{hickey2016signatures}
James~M Hickey, Sam Genway, and Juan~P Garrahan.
\newblock Signatures of many-body localisation in a system without disorder and
  the relation to a glass transition.
\newblock {\em Journal of Statistical Mechanics: Theory and Experiment},
  2016(5):054047, 2016.

\bibitem{huse2014phenomenology}
David~A Huse, Rahul Nandkishore, and Vadim Oganesyan.
\newblock Phenomenology of fully many-body-localized systems.
\newblock {\em Physical Review B}, 90(17):174202, 2014.

\bibitem{huveneers2013drastic}
Fran{\c{c}}ois Huveneers.
\newblock Drastic fall-off of the thermal conductivity for disordered lattices
  in the limit of weak anharmonic interactions.
\newblock {\em Nonlinearity}, 26(3):837, 2013.

\bibitem{imbrie2016many}
John~Z Imbrie.
\newblock On many-body localization for quantum spin chains.
\newblock {\em Journal of Statistical Physics}, 163(5):998--1048, 2016.

\bibitem{imbrie2016review}
John~Z Imbrie, Valentina Ros, and Antonello Scardicchio.
\newblock Review: Local integrals of motion in many-body localized systems.
\newblock {\em arXiv preprint arXiv:1609.08076}, 2016.

\bibitem{kagan1984localization}
Yu~Kagan and LA~Maksimov.
\newblock Localization in a system of interacting particles diffusing in a
  regular crystal.
\newblock {\em Zhurnal Eksperimental’noi i Teoreticheskoi Fiziki},
  87:348--365, 1984.

\bibitem{khemani2016obtaining}
Vedika Khemani, Frank Pollmann, and SL~Sondhi.
\newblock Obtaining highly excited eigenstates of many-body localized
  hamiltonians by the density matrix renormalization group approach.
\newblock {\em Physical Review Letters}, 116(24):247204, 2016.

\bibitem{kim2016localization}
Isaac~H Kim and Jeongwan Haah.
\newblock Localization from superselection rules in translationally invariant
  systems.
\newblock {\em Physical review letters}, 116(2):027202, 2016.

\bibitem{kjall2014many}
Jonas~A Kj{\"a}ll, Jens~H Bardarson, and Frank Pollmann.
\newblock Many-body localization in a disordered quantum ising chain.
\newblock {\em Physical review letters}, 113(10):107204, 2014.

\bibitem{luitz2015many}
David~J Luitz, Nicolas Laflorencie, and Fabien Alet.
\newblock Many-body localization edge in the random-field heisenberg chain.
\newblock {\em Physical Review B}, 91(8):081103, 2015.

\bibitem{mastropietro2016localization}
Vieri Mastropietro.
\newblock Localization in the ground state of an interacting quasi-periodic
  fermionic chain.
\newblock {\em Communications in Mathematical Physics}, 342(1):217--250, 2016.

\bibitem{nandkishore2014many}
Rahul Nandkishore and David~A Huse.
\newblock Many body localization and thermalization in quantum statistical
  mechanics.
\newblock {\em Annual Review of Condensed Matter Physics}, 6:15--38, 2015.

\bibitem{oganesyan2007localization}
Vadim Oganesyan and David~A Huse.
\newblock Localization of interacting fermions at high temperature.
\newblock {\em Physical Review B}, 75(15):155111, 2007.

\bibitem{pal2010many}
Arijeet Pal and David~A Huse.
\newblock Many-body localization phase transition.
\newblock {\em Physical review b}, 82(17):174411, 2010.

\bibitem{papic2015many}
Z~Papi{\'c}, E~Miles Stoudenmire, and Dmitry~A Abanin.
\newblock Many-body localization in disorder-free systems: The importance of
  finite-size constraints.
\newblock {\em Annals of Physics}, 362:714--725, 2015.

\bibitem{pino2016nonergodic}
Manuel Pino, Lev~B Ioffe, and Boris~L Altshuler.
\newblock Nonergodic metallic and insulating phases of josephson junction
  chains.
\newblock {\em Proceedings of the National Academy of Sciences},
  113(3):536--541, 2016.

\bibitem{poschel1993nekhoroshev}
J{\"u}rgen P{\"o}schel.
\newblock Nekhoroshev estimates for quasi-convex hamiltonian systems.
\newblock {\em Mathematische Zeitschrift}, 213(1):187--216, 1993.

\bibitem{potter2015universal}
Andrew~C Potter, Romain Vasseur, and SA~Parameswaran.
\newblock Universal properties of many-body delocalization transitions.
\newblock {\em Physical Review X}, 5(3):031033, 2015.

\bibitem{ros2015integrals}
V~Ros, M~M{\"u}ller, and A~Scardicchio.
\newblock Integrals of motion in the many-body localized phase.
\newblock {\em Nuclear Physics B}, 891:420--465, 2015.

\bibitem{schiulaz2013ideal}
Mauro Schiulaz and M~M{\"u}ller.
\newblock Ideal quantum glass transitions: many-body localization without
  quenched disorder.
\newblock {\em AIP Conf. Proc.}, 1610.

\bibitem{schiulaz2015dynamics}
Mauro Schiulaz, Alessandro Silva, and Markus M{\"u}ller.
\newblock Dynamics in many-body localized quantum systems without disorder.
\newblock {\em Physical Review B}, 91(18):184202, 2015.

\bibitem{seiringer2016decay}
Robert Seiringer and Simone Warzel.
\newblock Decay of correlations and absence of superfluidity in the disordered
  tonks--girardeau gas.
\newblock {\em New Journal of Physics}, 18(3):035002, 2016.

\bibitem{serbyn2013local}
Maksym Serbyn, Z~Papi{\'c}, and Dmitry~A Abanin.
\newblock Local conservation laws and the structure of the many-body localized
  states.
\newblock {\em Physical review letters}, 111(12):127201, 2013.

\bibitem{shapir1982localization}
Yonathan Shapir, Amnon Aharony, and A~Brooks Harris.
\newblock Localization and quantum percolation.
\newblock {\em Physical Review Letters}, 49(7):486, 1982.

\bibitem{singh2016localization}
Rajeev Singh and Efrat Shimshoni.
\newblock Localization due to interaction-enhanced disorder in bosonic systems.
\newblock {\em arXiv preprint arXiv:1610.01395}, 2016.

\bibitem{srednicki1994chaos}
Mark Srednicki.
\newblock Chaos and quantum thermalization.
\newblock {\em Physical Review E}, 50(2):888, 1994.

\bibitem{stolz2016many}
G{\"u}nter Stolz.
\newblock Many-body localization for disordered bosons.
\newblock {\em New Journal of Physics}, 18(3):031002, 2016.

\bibitem{van2015dynamics}
Merlijn van Horssen, Emanuele Levi, and Juan~P Garrahan.
\newblock Dynamics of many-body localization in a translation-invariant quantum
  glass model.
\newblock {\em Physical Review B}, 92(10):100305, 2015.

\bibitem{veselic2005spectral}
Ivan Veseli{\'c}.
\newblock Spectral analysis of percolation hamiltonians.
\newblock {\em Mathematische Annalen}, 331(4):841--865, 2005.

\bibitem{vosk2015theory}
Ronen Vosk, David~A Huse, and Ehud Altman.
\newblock Theory of the many-body localization transition in one-dimensional
  systems.
\newblock {\em Physical Review X}, 5(3):031032, 2015.

\bibitem{yao2014quasi}
N.~Y. Yao, C.~R. Laumann, J.~I. Cirac, M.~D. Lukin, and J.~E. Moore.
\newblock Quasi-many-body localization in translation-invariant systems.
\newblock {\em Phys. Rev. Lett.}, 117:240601, 2016.

\bibitem{vznidarivc2008many}
Marko {\v{Z}}nidari{\v{c}}, Toma{\v{z}} Prosen, and Peter Prelov{\v{s}}ek.
\newblock Many-body localization in the heisenberg x x z magnet in a random
  field.
\newblock {\em Physical Review B}, 77(6):064426, 2008.

\end{thebibliography}

\end{document}